\documentclass[12pt]{article}
\usepackage{bbm,fullpage}
\usepackage{bm}
\usepackage{amsmath,amssymb,amsthm,amsmath}
\usepackage{algorithm, pseudocode}
\usepackage{mathrsfs}
\usepackage{enumitem}
\usepackage[titles]{tocloft}
\usepackage{yfonts}
\usepackage{xcolor}
\usepackage[pdftex,                %
     bookmarks         = true,%     % Signets
    bookmarksnumbered = true,%     % Signets numérotés
    pdfpagemode       = None,%     % Signets/vignettes fermé à l'ouverture
    pdfstartview      = FitH,%     % La page prend toute la largeur
    pdfpagelayout     = SinglePage,% Vue par page
    colorlinks        = true,%     % Liens en couleur
   citecolor         =blue,
    urlcolor          = magenta,%  % Couleur des liens externes
    pdfborder         = {0 0 0}%   % Style de bordure : ici, pas de bordure
    ]{hyperref}%                   % Utilisation de HyperTeX
 
\usepackage{bbm}
\usepackage{amsmath}
\usepackage{easybmat}
\usepackage{multirow,bigdelim}

\usepackage[indexonlyfirst,ucmark,toc]{glossaries}

\definecolor{ao(english)}{rgb}{0.0, 0.5, 0.0}

\def\softO{\ensuremath{{O}{\,\tilde{ }\,}}}

\DeclareBoldMathCommand{\bnu}{\nu}

\def\Sc{\ensuremath{\mathcal{S}}}

\def\Tb {\ensuremath{\mathbb{T}}}

\def\KKbar {\ensuremath{\overline{\mathbf{K}}}}

\def\KK {\ensuremath{\mathbf{K}}}

\def\d {\ensuremath{\mathbf{d}}}

\def\G {\ensuremath{\mathrm{GL}}}

\def\jac{\ensuremath{{\rm Jac}}}

\def\scrR{\ensuremath{\mathscr{R}}}

\def\e {\ensuremath{\bm{e}}}

\def\b {\ensuremath{\mathbf{b}}}

\def\f {\ensuremath{\mathbf{f}}}

\def\L {\ensuremath{\mathbf{L}}}

\def\G {\ensuremath{\mathbf{G}}}

\def\Y {\ensuremath{\mathbf{Y}}}
\def\H {\ensuremath{\mathbf{H}}}

\def\U {\ensuremath{\mathbf{U}}}

\def\b_eta{\mbox{\boldmath$\eta$}}

\def\bzeta{\mbox{\boldmath$\zeta$}}
\def\bkappa{\mbox{\boldmath$\kappa$}}

\def\bh{\mbox{\boldmath$h$}}

\def\bell{\mbox{\boldmath$\ell$}}

\def\mat#1{\ensuremath{\bm{#1}}}

\def\KKbar {\ensuremath{\overline{\mathbf{K}}}}

\DeclareBoldMathCommand{\bq}{q}
\DeclareBoldMathCommand{\bp}{p}
\DeclareBoldMathCommand{\y}{y}
\DeclareBoldMathCommand{\bM}{M}
\DeclareBoldMathCommand{\bD}{D}
\DeclareBoldMathCommand{\bB}{B}
\DeclareBoldMathCommand{\bA}{A}
\DeclareBoldMathCommand{\bC}{C}
\DeclareBoldMathCommand{\bX}{X}
\DeclareBoldMathCommand{\bZ}{Z}
\DeclareBoldMathCommand{\bi}{i}
\DeclareBoldMathCommand{\bk}{k}
\DeclareBoldMathCommand{\bY}{Y}
\DeclareBoldMathCommand{\bn}{n}
\DeclareBoldMathCommand{\bbb}{b}
\DeclareBoldMathCommand{\d}{d}
\DeclareBoldMathCommand{\ba}{a}
\DeclareBoldMathCommand{\bc}{c}
\DeclareBoldMathCommand{\bj}{j}
\DeclareBoldMathCommand{\be}{e}
\DeclareBoldMathCommand{\bh}{h}
\DeclareBoldMathCommand{\br}{r}
\DeclareBoldMathCommand{\bu}{u}
\DeclareBoldMathCommand{\bv}{v}
\DeclareBoldMathCommand{\bx}{x}
\DeclareBoldMathCommand{\bz}{z}
\DeclareBoldMathCommand{\by}{y}
\DeclareBoldMathCommand{\bw}{w}
\DeclareBoldMathCommand{\c}{c}
\DeclareBoldMathCommand{\f}{f}
\DeclareBoldMathCommand{\g}{g}
\DeclareBoldMathCommand{\h}{h}
\DeclareBoldMathCommand{\x}{x}
\DeclareBoldMathCommand{\bell}{\ell}
\DeclareBoldMathCommand{\bbeta}{\beta}
\DeclareBoldMathCommand{\bchi}{\chi}
\DeclareBoldMathCommand{\balpha}{\alpha}
\DeclareBoldMathCommand{\bgamma}{\gamma}
\DeclareBoldMathCommand{\bkappa}{\kappa}
\DeclareBoldMathCommand{\btheta}{\vartheta}
\DeclareBoldMathCommand{\bvartheta}{\vartheta}
\DeclareBoldMathCommand{\bzeta}{\zeta}

\def\wdeg{\ensuremath{\mathrm{wdeg}}}

\def\crit{\ensuremath{W}}

\def\EQS{\bm{f}}
\def\eq{{f}}
\def\MAT{{\G}}
\def\mat#1{{\bf #1}}
\def\gmat{{g}}

\def\symmrep{symmetric representation}
\def\Upsymmrep{Symmetric representation}
\def\Clambda{{\cal C}_\lambda}
\def\ClambdaS{\Clambda^{\rm strict}}
\def\OurOmega{\omega}

\def\mycomment#1{}

\title{Computing critical points for invariant algebraic systems}
\author{Jean-Charles Faug\`ere\thanks{Inria, Sorbonne Universit\'e, CNRS,
    Laboratoire d'Informatique de Paris 6, \'Equipe PolSys, CryptoNext Security,
    4 place Jussieu, F-75252, Paris Cedex 05, France, email:{\text{Jean-Charles.Faugere@inria.fr.}}},   George Labahn\thanks{Cheriton School of Computer Science,
        University of Waterloo, Waterloo ON, Canada N2L 3G1, emails:{\text{\{glabahn, eschost, txvu\}@uwaterloo.ca.}}}, 
        \and Mohab Safey El Din\thanks{Sorbonne Universit\'e, CNRS,
    Laboratoire d'Informatique de Paris 6, \'Equipe PolSys, 4 place Jussieu, F-75252, Paris Cedex 05, France, email:{\text{Mohab.Safey@lip6.fr.}}}, 
        \'Eric Schost$^\dagger,$  Thi Xuan Vu$^{* \ddagger}$
}

\date{}

\newcommand{\rvline}{\hspace*{-\arraycolsep}\vline\hspace*{-\arraycolsep}}

\newtheorem{definition}{Definition}
\numberwithin{definition}{section}
\newtheorem{theorem}[definition]{Theorem}
\newtheorem{corollary}[definition]{Corollary}
\newtheorem{proposition}[definition]{Proposition}
\newtheorem{lemma}[definition]{Lemma}
\newtheorem{remark}[definition]{Remark}
\newtheorem{example}[definition]{Example}

\begin{document}
\maketitle
\begin{abstract} 
  Let $\KK$ be a field and $\phi$, $\f = (f_1, \ldots, f_s)$ in $\KK[x_1, \dots,
  x_n]$ be multivariate polynomials (with $s < n$) invariant under the action of
  $\Sc_n$, the group of permutations of $\{1, \dots, n\}$. We consider the
  problem of computing the points at which $\f$ vanish and the Jacobian matrix
  associated to $\f, \phi$ is rank deficient provided that this set is finite.
 
   We exploit the invariance properties of the input to split the
   solution space according to the orbits of $\Sc_n$.  This allows us
   to design an algorithm which gives a triangular description of the
   solution space and which runs in time polynomial in $d^s$,
   ${{n+d}\choose{d}}$ and $\binom{n}{s+1}$ where $d$ is the maximum
   degree of the input polynomials. When $d,s$ are fixed, this is
   polynomial in $n$ while when $s$ is fixed and $d \simeq n$ this
   yields an exponential speed-up with respect to the usual polynomial
   system solving algorithms.
 \end{abstract}

%\todo{spell-check}
%\todo{check that the refs look OK}

%\input{abstract}

\section{Introduction}

Our main motivation in this paper is the problem of finding the
critical points of a polynomial map $\phi$ restricted to an algebraic
set $V(\f)$, where $\f=(f_1, \dots, f_s)$ and $\phi$ come from the
multivariate polynomial ring $\KK[x_1, \ldots, x_n]$, with $\KK$ a
field of characteristic zero.  The problem of computing such points
appears in many application areas including for example polynomial optimization and real
algebraic geometry.
  
In our case we consider the closely related problem of computing a
description of the set $\crit(\phi, \f)$ defined by the following
equations:
\begin{equation}\label{intro:one}
   \langle f_1, \dots,  f_s \rangle + \langle M_{s+1}
  (\jac(\f, \phi))\rangle
\end{equation}
where, $\jac(\f, \phi)$ is the Jacobian matrix of $(f_1, \dots, f_s,
\phi)$ with respect to $(x_1, \dots, x_n)$, and $M_r( \MAT )$ denotes
the set of all $r$-minors of a matrix $\MAT$. If we assume that the
Jacobian matrix $\jac(\f)$ has full rank $s$ at any point of $V(\f)$,
then, the Jacobian criterion~\cite[Theorem~16.19]{Eisenbud95} implies
that the algebraic set $V(\f)$ is smooth and $(n-s)$-equidimensional,
and that $\crit(\phi, \f)$ is indeed the set of critical points of
$\phi$ on $V(\f)$. 

When $\phi$ is linear, there exist algorithms for determining critical
points using $d^{O(n)}$ operations in
$\KK$~\cite[Section~14.2]{Basu2006}. More precisely, using Gr\"obner
basis techniques, the paper~\cite[Corollary 3]{Faug12} establishs that, if
the polynomials $f_1,\dots,f_s$ are generic enough of degree $d$, then
this computation can be done using
\[
  O\Big( {{n+D_{{\rm reg}}} \choose {n}}^{\OurOmega}  + n \left (d^s \, (d-1)^{n-s}
 \, {{n-1}
    \choose {s-1}}\right )^{3}\Big)
\] 
operations in $\KK$. Here $D_{{\rm reg}} = d(s-1) + (d-2)n + 2$, and
$\OurOmega$ is the exponent of multiplying two $(n \times n)$-matrices
with coefficients in $\KK$ (see~\cite{Spa14} for a generalization to
systems with mixed degrees).

In this paper, we consider the important case where the polynomials
$f_1,\dots,f_s$ and $\phi$ are all invariant under the action of the
symmetric group $\Sc_n$. As we will show later, the set $\crit(\phi,
\f)$ is then also invariant under $\Sc_n$.

There has been considerable work on solving symmetric algebraic
systems. Indeed, while it is always possible to compute the Gr\"obner
basis of a set of symmetric polynomials, symmetries of the initial
system are lost during the computation. In~\cite{Colin97}, for a
finite symmetry group, Colin proposed to use primary and secondary
invariants~\cite{Sturmfels93} to reformulate the problem. For the
particular case of $\Sc_n$-invariant equations, in~\cite{Faugere09},
the authors compute a SAGBI-Gr\"obner basis in the ring $\KK[e_1,
  \dots, e_n]$, where $e_i$ is a variable corresponding to $i$-th
elementary symmetric polynomial $\eta_i$ in $(x_1, \dots, x_n)$.
However, even if $f_1,\dots,f_s$ and $\phi$ are $\Sc_n$-invariant, the
equations in~\eqref{intro:one} are usually not invariant, so these technique
cannot be directly applied to our problem.

It is possible to prove that the system of equations in
\eqref{intro:one} is {\em globally invariant}: for all $\sigma \in
\Sc_n$, and any $g$ among either $f_1,\dots,f_s$ or the $(s+1)$-minors
of $\jac(\f, \phi)$, either $\sigma(g)$ or $-\sigma(g)$ belongs again
to the same set of equations. This implies that $\crit(\phi, \f)$ is
$\Sc_n$-invariant, as we claimed above. As an example, with $n=3$ and $s=1$, in order to determine
the critical points of $\phi = x_1x_2x_3 - 3 x_1 - 3x_2 - 3x_3$ over
the sphere defined by $f = x_1^2 + x_2^2 + x_3^2 - 6$, one has to 
solve the globally invariant set of equations defined by 
$$
\{ f =0\, , \,x_1^2 x_3 - x_2^2 x_3 - 3x_1 +3 x_2\,=0 , \,
x_1^2 x_2 - x_2 x_3^2 - 3x_1 +3 x_3\,=0, \,
x_1 x_2^2 -  x_1 x_3^2 - 3x_2 +3 x_3\, =0\}.
$$ 
For such systems, following~\cite{FaHePh03}, the authors
in~\cite{FaugereJules12} used divided differences to construct a new
system which is $\Sc_n$-invariant. Our work is inspired by this
reference, but the specific type of the equations that we solve,
involving minors of a Jacobian matrix, requires us to extend the work
from~\cite{FaugereJules12} (in addition, no complexity analysis is
given in that reference).

The global invariance property allows us to split the set
$W=\crit(\phi, \f)$ into orbits under the action of the symmetric
group. The size of the orbit of a point in $W$ will depend on the
number of pairwise distinct coordinates of that point. For example,
for $f$ and $\phi$ as above, the points $(2,1,1), (0, \sqrt{3},
\sqrt{3}), (-2,-1,-1)$ are solutions with three elements in their
respective $\Sc_3$-orbits, while the point $(\sqrt{2}, \sqrt{2},
\sqrt{2})$ is also a solution, with only one point in its orbit (this
is the whole decomposition of $W$ into orbits). To devise a fast
algorithm, the different sizes of orbits needs to be taken into
consideration.  This phenomenon is to be expected for systems such
as~\eqref{intro:one}, but is not discussed for the particular family
of equations in~\cite{FaugereJules12} (on the other hand, that
reference takes into consideration further properties of the family of
equations considered therein).
 
The structure of these orbits is determined by the number of pairwise
distinct coordinates of the points they contain. To study them, we
make use of partitions of $n$.  A sequence $\lambda = (n_1^{\ell_1} \,
n_2^{\ell_2}\, \dots\, n_r^{\ell_r})$, with the $\ell_i$ and $n_i$
positive integers and $n_1 < \cdots < n_r$, is called a {\it
  partition} of $n$ if $n_1 \ell_1 + n_2 \ell_2 + \cdots + n_r \ell_r
= n$. Partitions of $n$ will be used to parameterize orbits, with
$\lambda$ as above parameterizing those points in $W$ having $\ell_1$
distinct sets of $n_1$ equal coordinates, $\ell_2$ distinct sets of
$n_2$ equal coordinates and so on. We will write $W_\lambda$ for the
set of such orbits contained in $W$, so that $W$ is the disjoint union
of all $W_\lambda$, for all partitions $\lambda$ of $n$.

For instance, for the $\phi$ and $f$ mentioned previously, our
algorithm will determine that the set $W_{(1^3)}$ of orbits parameterized by
$\lambda=(1^3)$, which corresponds to the orbits with all distinct
coordinates $(\xi_1,\xi_2,\xi_3)$, is equal to the zero set of
\[ 
(f,\ -4,\ -2(x_1+x_2+x_3),\ 2(x_1^2+x_2^2+x_3^2) + 8(x_1x_2+x_2x_3+x_1x_3) - 36)
\] 
(and so $W_{(1^3)}$ is empty, as we saw above). The set
$W_{(1^1\, 2^1)}$ of orbits parameterized by $\lambda=(1^1\, 2^1)$, that is,
orbits of points of the form $(\xi_1,\xi_2,\xi_2)$, with $\xi_1 \ne
\xi_2$, is the orbit of the zero set of $$( x_1^2 + 2x_2^2 -
6,\ x_2^2+x_1x_2-3,\ x_2-x_3),$$ where the first component is $f$
restricted to the hyperplane $x_2=x_3$. In particular,
$W_{(1^1\, 2^1)}$ is the union of the orbits of the points $(2,1,1),
(0, \sqrt{3}, \sqrt{3}), (-2,-1,-1)$ seen above.

% contribution 2
In this paper we provide a procedure to determine invariant
polynomials that describe these $\Sc_n$-orbits. For an orbit
parameterized by the partition $\lambda =(n_1^{\ell_1} \,
n_2^{\ell_2}\, \dots\, n_r^{\ell_r})$, we work with points which have
distinct coordinates $(\xi_{1,1}, \ldots, \xi_{1,\ell_1}, \xi_{2, 1},
\ldots, \xi_{2,\ell_2}, \ldots , \xi_{r, 1}, \ldots, \xi_{r,
  \ell_r})$, so that instead of $n$ coordinates, there are only $\ell
= \ell_1 + \cdots + \ell_r$ distinct coordinates for points in this
orbit. Then, invariance under of $W$ permutations implies that single
distinct points are permuted, groups of two points are permuted,
etc. This will allow us to work with polynomials in $ \KK[\bm
  e_1,\dots,\bm e_r]=\KK[e_{1, 1}, \ldots, e_{1,\ell_1},e_{2, 1},
  \ldots, e_{2,\ell_2}, \dots, e_{r, 1}, \ldots, e_{r,\ell_r}]$, in
order to represent a certain ``compressed'' image $W'_\lambda \subset
\KKbar{}^\ell$ of $W_\lambda$; here, $e_{i, 1}, \ldots, e_{i,\ell_i}$
are variables standing for the elementary symmetric polynomial in
$\ell_i$ indeterminates and $\KKbar$ is an algebraic closure of $\KK$.
In our running example, for $\lambda =(1^1\,2^1)$, we have $\ell=2$
and $W'_{(1^1\, 2^1)}$ is the set $\{(2,1),\,(0,\sqrt 3),\,(-2,-1)\}$.

Throughout the paper, we will assume that $W$, and thus all
$W_\lambda$ and $W'_\lambda$, are finite. Then, for $\lambda$ as
above, the cardinality of $W'_\lambda$ is smaller than that of
$W_\lambda$ by a factor $$\gamma_\lambda = {{n}\choose {n_1, \dots,
    n_1, \dots, n_r, \dots, n_r}},$$ where each $n_i$ is repeated
$\ell_i$ times.  Altogether, if $d$ is the maximum of the degrees of
the input of polynomials, then we will prove some bounds, which will
be denoted by $\mathfrak{c}_\lambda$, on the cardinality of the finite
set $W'_\lambda$; we will see that, in practice, each of the
$\mathfrak{c}_\lambda$ provides an accurate bound on the cardinality
of $W'_\lambda$. The sum of the $\mathfrak{c}_\lambda$'s then gives us
an upper bound on the size of the output of our main algorithm.  We
did not find a closed formula for this sum, but we can prove that it
is bounded above by
\begin{equation}\label{intro:two}
\mathfrak{c} = d^s\, {{n+d-1} \choose {n}}.
\end{equation}
We will see that, in practice, this is a rather rough upper bound but
in several cases, it
compares well to the upper bound
\begin{equation} \label{eq:naive}
\tilde{\mathfrak{c}} = d^s \, (d-1)^{n-s}\, {n \choose s}
\end{equation}
from Nie and Ranestad~\cite[Theorem~2.2]{NieRan09} on the size of $W$.
For example, when $d=2$, we have $\mathfrak{c} = 2^s
(n+1)$ while $\tilde{\mathfrak{c}} = 2^s {n \choose s}$.  More
generally, when $d$ and $s$ are fixed, $\mathfrak{c}$ is polynomial in
$n$ (since it is bounded above by $d^s(n+d-1)^d$) while
$\tilde{\mathfrak{c}}$ is exponential in $n$ (since it is greater than
$(d-1)^{n}$). When $s$ is fixed and $d = n$, $\mathfrak{c}$ is 
$n^{O(1)} 2^n$, whereas $\tilde{\mathfrak{c}}$ is  $n^{O(1)}
(n-1)^{n-s}$.

% why this contribution is significant
In view of this discussion, our algorithm will naturally compute
descriptions of the sets $W'_\lambda$ rather than $W_\lambda$ (we will
also explain how one would recover the later knowing the former).
There are a number of ways to represent algebraic sets; in our case we
make use of a representation based on univariate polynomials.  In
particular, if $Y \subset \KKbar{}^m$ is a zero-dimensional variety
defined by polynomials in $\KK[z_1,\dots,z_m]$, then a {\em
  zero-dimensional parametrization} $\scrR = ((q, v_1, \dots, v_m),
\mu)$ of $Y$ consists of
\begin{itemize}
\item[(i)] a squarefree polynomial $q$ in $\KK[y]$, with $y$ a new
  indeterminate and $\deg(q) = |Y|$,
     
\item[(ii)] polynomials $(v_1, \dots, v_m)$ in $\KK[y]$ with 
  $\deg(v_i) < \deg(q)$ for all $i$, and satisfying 
  $Y = \{(v_1(\tau), \dots, v_m(\tau)) \in \KKbar{}^m \, | \, q(\tau) = 0\}$,
\item[(iii)] a vector $\mu = (\mu_1,\dots,\mu_m)$ in $\KK^m$ such that
  $\mu_1 v_1 + \cdots + \mu_m v_m = y$.
\end{itemize}
When these conditions hold, we write $Y=Z(\scrR)$.

The last condition says that the roots of $q$ are the values taken by
the linear form $\mu_1 z_1 + \cdots + \mu_m z_m$ on $Y$. In
particular, this linear form takes pairwise distinct values on the
points of $Y$. This representation was first introduced in the works
of Kronecker and K\"{o}nig~\cite{Kronecker82} and has been widely used
in computer algebra~\cite{ABRW96, GiMo89,
  GiHeMoMoPa98,GiHeMoPa95,GiLeSa01, Rouillier99}.
The output of our algorithm will thus be a collection of
zero-dimensional parameterizations, one for each of the sets
$W'_\lambda$; we will call such a data structure a {\em {\symmrep}} of
$W$ (precise definitions are in Section~\ref{sec:invariant}). 

However, rather than using Gr\"obner bases to compute such
descriptions, we will use a {\it symbolic homotopy continuation}, so
as to control precisely the cost of the algorithm. Homotopy
continuation has become a foundational tool for numerical algorithms
while the use of symbolic homotopy continuation algorithms is more
recent. Such algorithms first appeared in~\cite{Bomp04, Heintz00}, for
general inputs, and later for sparse~\cite{Jer09, HeJeSa10, Maria13,
  HeJeSa14} and multi-homogeneous systems~\cite{SaSc18, HeJeSaSo02,
  HSSV18}.

In our case we can make use of a recent sparse symbolic homotopy
method given in~\cite{sparse-homotopy} specifically designed to handle
determinantal systems over weighted polynomial rings, that is,
multivariate polynomial rings where each variable has a weighted
degree, which is a positive integer. These domains arise naturally for
our orbits: the domain arising from an orbit parameter $\lambda$ has
variables $e_{i,k}$ which are defined corresponding to elementary
symmetric polynomials $\eta_{i, k}$; since $\eta_{i,k}$ has degree
$k$, the variable $e_{i,k}$ will naturally be assigned weight $k$.

\begin{theorem}\label{thm:one}
  Suppose $\f = (f_1, \dots, f_s)$ and $\phi$ are $\Sc_n$-invariant
  polynomials in $\KK[x_1, \dots, x_n]$, with degree at most $d \ge
  2$, and suppose that $W=\crit(\phi, \f)$ is finite.  There exists a
  randomized algorithm that takes $\f,\phi$ as input and outputs a
  {\symmrep} for the set $W$, and whose runtime is polynomial in
  $d^s,\, \binom{n+d}{d},\, {n \choose {s+1}}$. The total number of
  points described by the output is at most $d^s\, {{n+d-1} \choose
    {n}}$.
\end{theorem}
\noindent Note that the runtime is polynomial in the bound we give on
the output size, as well as the number ${n \choose {s+1}}$ of maximal
minors in the matrix $\jac(\f,
\phi)$. Section~\ref{sec:weightedhomotopy} gives a more precise estimate
on the runtime of the algorithm.

We use standard notions and notations of commutative algebra and
algebraic geometry which can be found for example in~\cite{CLO07,
  Eisenbud95}.  We will assume that the reader is familiar with
concepts such as {\it dimension}, {\it Zariski topology}, {\it
  equidimensional algebraic set} and the {\it degree} of an algebraic
set, with definitions found in~\cite{CLO07, Eisenbud95}.

% Plan of the paper.
The remainder of the paper is organized as follows. In the next
section, we provide several properties of invariant polynomials and
discuss in detail the sets $W_\lambda$ and $W'_\lambda$ mentioned
above.  Section~\ref{sec:crit} contains our main algorithm, called
${\sf Critical\_Points\_Per\_Orbit}$ and includes a proof of
correctness. The runtime of this algorithm is analysed in
Section~\ref{sec:weightedhomotopy}, finishing the proof of
Theorem~\ref{thm:one}.  Experiments to validate our new algorithm is
given in Section~\ref{sec:experiments} followed by a section which
gives topics for future research. The latter section also includes a
discussion on how our results can decide emptiness of
$\Sc_n$-invariant algebraic sets over a real field. The appendices
include a proof of two technical propositions.

%\input{introduction}
%%\input{preliminaries}
%%%%%%%%%%%%%%%%%%%%%%%%%%%%%%%%%%%%%%%%%%%%%%%%%%%%%%%%%%%%
%%%%%%%%%%%%%%%%%%%%%%%%%%%%%%%%%%%%%%%%%%%%%%%%%%%%%%%%%%%%
%%%%%%%%%%%%%%%%%%%%%%%%%%%%%%%%%%%%%%%%%%%%%%%%%%%%%%%%%%%%

\section{Partitions and distinct coordinates of $\Sc_n$-invariants}\label{sec:invariant}

One of our key observations, formalized in
the next section, is that the special nature of our set of critical
points allows us to split $\crit(\phi, \f)$ into subsystems defined
by the orbits of the symmetric group $\Sc_n$.  

More precisely, in this paper an {\em orbit} is a set of the form
$\Sc_n(\bm\xi)$, for some point $\bm\xi$ in $\KKbar{}^n$, that is, it
is the set of all $\Sc_n$-conjugates of $\bm\xi$. As mentioned in the
introduction, the size of an orbit $\Sc_n(\bm\xi)$ will depend on the
number of pairwise distinct coordinates of $\bm\xi$. For example, with
$n=3$, a point of the form $(\xi_1,\xi_2,\xi_2)$ will have an orbit of
size $3$, unless we have $\xi_1 = \xi_2$ (in which case the orbit has
size $1$).

As a result, we need to consider the separation of distinct
coordinates in an orbit, which is what we do in this section. We do
this through a discussion of the geometry of (finite)
$\Sc_n$-invariant subsets of $\KKbar{}^n$ and the data structures we
can use to represent them.  Much of what follows is preliminary for our
description of orbits presented in the next section.

%%%%%%%%%%%%%%%%%%%%%%%%%%%%%%%%%%%%%%%%%%%%%%%%%%%%%%%%%%%%

\subsection{Partitions}\label{ssec:partitions}

Partitions play a major role in describing our orbits. In this
subsection, we gather the basic definitions of partitions and of a few
notions attached to them, which will be used throughout this section.

A sequence $\lambda = (n_1^{\ell_1} \, n_2^{\ell_2}\, \dots\,
n_r^{\ell_r})$, with $\ell_i$'s and $n_i$'s positive integers and $n_1
< \cdots < n_r$, is called a {\it partition} of $n$, sometimes denoted
by $\lambda\vdash n$, if $n_1 \ell_1 + n_2 \ell_2 + \cdots + n_r
\ell_r = n$. The number $\ell = \sum_{i=1}^r \, \ell_i$ is called the
\textit{length} of the partition $\lambda$. We remark that to a
partition such as $\lambda = (n_1^{\ell_1} \, n_2^{\ell_2}\, \dots\,
n_r^{\ell_r})$ we can associate (in a one-to-one manner) the ordered
list $(n_1,\dots,n_1,\dots,n_r,\dots,n_r)$, with each $n_i$ repeated
$\ell_i$ times.

We will make use of the {\em refinement order} on partitions.  To describe this we first
need to define the union of partitions: if $\lambda$ and $\lambda'$ are
partitions of  $a$ and $a'$, respectively, then  $\lambda \cup
\lambda'$ is the partition of $a+a'$ whose ordered list is obtained
by merging those of $\lambda$ and $\lambda'$.  Then, consider two
partitions $\lambda = (n_1^{\ell_1} \, n_2^{\ell_2}\, \dots\,
n_r^{\ell_r})$ and $\lambda' = (m_1^{k_1} \, m_2^{k_2}\, \dots\,
m_s^{k_s})$ of the same integer $n$. As
in~\cite[p.~103]{Macdonald1998} (or e.g.~\cite[p.~16]{Birkhoff67}), we
write $\lambda \le \lambda'$, and we say that $\lambda$ {\em refines}
$\lambda'$, if $\lambda$ is the union of some partitions
$(\lambda_{i,j})_{1 \le i \le s, 1 \le j \le k_i}$, where $\lambda_{i,j}$ is
a partition of $m_i$ for all $i,j$.

\begin{example}
  For the partitions of $n=3$, we have $(1^3) \le (1^1 2^1) \le (3^1)$.
\end{example}

Let $\lambda = (n_1^{\ell_1} \, n_2^{\ell_2}\, \dots\,
n_r^{\ell_r})$ be a partition of $n$ having length $\ell$. For $k=1,\dots,r$, we will
denote by $\bZ_k = (z_{k, 1}, \dots, z_{k, \ell_k})$ a sequence of
$\ell_k$ indeterminates. When convenient, we will also index the
entire sequence of indeterminates
$(\bZ_1,\dots,\bZ_r)=(z_{1,1},\dots,z_{r,\ell_r})$ as
$(z_1,\dots,z_\ell)$, so that $z_1=z_{1,1},
\dots,z_\ell=z_{r,\ell_r}$. From this point of view, introducing
$\tau_0=0$ and $\tau_k = \sum_{i=1}^k \ell_i$, for $k=1,\dots,r$, 
any index $i$ in $1,\dots,\ell$ can be
written uniquely as $i=\tau_{k-1} + u$, for some $k$ in $1,\dots,r$ 
and $u$ in $1,\dots,\ell_k$. Thus, the
indeterminates $z_{k, 1}, \dots, z_{k, \ell_k}$ are numbered
$z_{\tau_{k-1}+1},\dots,z_{\tau_{k}}$, with $\tau_r=\ell$.

We will let $\Sc_\lambda$ be the group
$$\Sc_\lambda = \Sc_{\ell_1} \times \cdots \times \Sc_{\ell_r}.$$
$\Sc_\lambda$ acts naturally on $\KK[\bZ_1, \dots, \bZ_r]$, and we
will denote by $\KK[\bZ_1, \dots, \bZ_r]^{\Sc_\lambda}$ the
$\KK$-algebra of $\Sc_\lambda$-invariant polynomials. Note that
$\Sc_\lambda$ can be seen as a subgroup of the permutation group
$\Sc_\ell$ of $\{1,\dots,\ell\}$, where $\Sc_{\ell_1}$ acts on the
  first $\ell_1$ indices, $\Sc_{\ell_2}$ acts on the next $\ell_2$
  ones, etc.

Finally, for $i=1,\dots,r$, we will let
$\bm\eta_i=(\eta_{i,1},\dots,\eta_{i,\ell_i})$ denote the vector of
elementary symmetric polynomials in variables $\bZ_i$, where
$\eta_{i,j}$ has degree $j$ for all $i,j$.

%%%%%%%%%%%%%%%%%%%%%%%%%%%%%%%%%%%%%%%%%%%%%%%%%%%%%%%%%%%%

\subsection{$\Sc_\lambda$-invariant polynomials: the {\sf Symmetric\_Coordinates} algorithm}

Let $\lambda = (n_1^{\ell_1} \, n_2^{\ell_2}\, \dots\, n_r^{\ell_r})$
be a partition of $n$ having length $\ell$, and, for $i=1,\dots,r$,
let $\bm e_i = (e_{i, 1}, \dots, e_{i, \ell_i})$ be a set of $\ell_i$
new variables. Then, by the fundamental theorem of symmetric
polynomials~\cite[Theorem~3.10.1]{DeKe02}, for any $f$ in $\KK[\bZ_1,
  \dots, \bZ_r]^{\Sc_\lambda}$, there exists a unique $\bar f$ in
$\KK[\bm e_1, \dots, \bm e_r]$ with
\begin{equation}\label{eqdef:fbar}
  f(\bZ_1,\dots,\bZ_r) = \bar f(\bm{\eta}_1,\dots,\bm{\eta}_r),
\end{equation}
for $\bm\eta_1,\dots,\bm\eta_r$ as defined in the previous subsection.
We will need a quantitative version of this existence result,
which gives an estimate on the cost of computing $\bar f$ from $f$.
\begin{lemma} \label{elem_slp}
  There exists an algorithm ${\sf Symmetric\_Coordinates}(\lambda, f)$
  which, given a partition $\lambda$ of $n$ and $f$ of degree at most
  $d$ in $\KK[\bZ_1, \dots, \bZ_r]^{\Sc_\lambda}$, returns $\bar f$
  such that $f = \bar f(\bm{\eta}_1,\dots,\bm{\eta}_r)$, using
  $\softO( {\ell + d \choose d}{}^2 )$ operations in $\KK$.\footnote{ Throughout this paper we use $\softO(\cdot)$ to
indicate that polylogarithmic factors are omitted, that is, $f$ is
$\softO(g)$ if there exists a constant $k$ such that $f$ is $O(g \,
\log^k(g))$.}
\end{lemma}
\begin{proof}
  Algorithm {\sf Symmetric\_Coordinates} is a slight generalization of
  the procedure described in the proof of Bl\"aser and Jindal's
  algorithm \cite[Theorem 4]{BlaserJindal18}, which was written only
  for the case of $r=1$, and for polynomials represented as
  straight-line programs.

  The key to the algorithm is the following.  Assume we know an
  integral domain $\bm{L}$ containing $\KK[\bm e_1, \dots, \bm e_r]$,
  and vectors $\bm{\zeta}_1,\dots,\bm{\zeta}_r$ of elements in $\bm{L}$,
  where for each $i$, $\bm{\zeta}_i = (\zeta_{i,1},\dots,\zeta_{i,\ell_i})
  \in \bm{L}^{\ell_i}$ are the $\ell_i$ pairwise distinct roots of
  \[
  P_i(T) = T^{\ell_i} - (e_{i,1} + \rho_{i,1}) T^{\ell_i - 1} + \cdots
  + (-1)^{\ell_i} \, (e_{i,\ell_i} + \rho_{i,\ell_i}), 
  \]
  and where $\rho_{i,1},\dots,\rho_{i,\ell_i}$ are the elementary
  symmetric polynomials evaluated at $1,\dots,\ell_i$. Then, $\bar f$
  satisfies $\bar f(e_{1,1} + 
  \rho_{1,1},\dots,e_{r,\ell_r}+\rho_{r,\ell_r}) =
  f(\bm{\zeta}_1,\dots,\bm{\zeta}_r)$.
  
  As in Bl\"aser and Jindal's algorithm, we take for $\bm{L}$ a ring
  of multivariate power series, namely $\bm{L}=\KK[[\bm e_1,\dots,\bm
      e_r]]$.  Our construction, involving the shifts by
  $(\rho_{1,1},\dots,\rho_{r,\ell_r})$ shows that at $\bm
  e_1=\cdots=\bm e_r=0$, $P_i(T)$ factors as $(T-1)\cdots
  (T-\ell_i)$. 

  Applying Newton's iteration, we deduce the existence of the
  requested power series roots
  $\bm{\zeta}_i=(\zeta_{i,1},\dots,\zeta_{i,\ell_i})$. In order to
  obtain the polynomial $\bar f$, we only need truncations of these
  roots at precision $d$. For $i=1,\dots,r$,  we can obtain the
  truncation of $\bm{\zeta}_i$ using $\softO(\ell_i {\ell_i + d
    \choose d})$ operations in $\KK$, where the factor ${\ell_i + d
    \choose d}$ accounts for the cost of multivariate power series
  arithmetic~\cite{LeSc03}. Taking all $i$'s into account, this adds
  up to $\softO(\ell {\ell + d \choose d})$ arithmetic operations.

  We then evaluate $f$ at these truncated power series. Since $f$ has
  degree at most $d$, this can be done using $O( {\ell + d \choose
    d})$ $(+,\times)$ operations on $\ell$-variate power series
  truncated in degree $d$, for a total of $\softO( {\ell + d \choose
    d}{}^2)$ operations in $\KK$.  This gives us $\bar f(e_{1,1} +
  \rho_{1,1},\dots,e_{r,\ell_r}+\rho_{r,\ell_r})$. We then apply the
  translation $(e_{i,j})_{i,j} \leftarrow (e_{i,j}-\rho_{i,j})_{i,j}$
  in order to obtain the polynomial $\bar f$, also at a cost of
  $\softO( {\ell + d \choose d}{}^2)$ operations in $\KK$: through
  successive multiplications, we incrementally compute the translates
  of all monomials of degree up to $d$ and then, before combining,
  using the coefficients of $\bar f(e_{1,1} +
  \rho_{1,1},\dots,e_{r,\ell_r}+\rho_{r,\ell_r})$.
\end{proof}

%%%%%%%%%%%%%%%%%%%%%%%%%%%%%%%%%%%%%%%%%%%%%%%%%%%%%%%%%%%%

\subsection{$\Sc_\lambda$-equivariant polynomials: the {\sf Symmetrize} algorithm}
\label{divided-diffs}

As before we let $\lambda = (n_1^{\ell_1} \, n_2^{\ell_2}\, \dots\,
n_r^{\ell_r})$  denote a partition of $n$ of length $\ell = \sum_{i=1}^r \,
\ell_i$.  The aim of this subsection is to define {\em
  $\Sc_\lambda$-equivariant} systems of polynomials and give a
detailed description of an algorithm, called ${\sf Symmetrize}$, that
turns an $\Sc_{\lambda}$-equivariant system into one which is
$\Sc_{\lambda}$-invariant.

Consider a sequence of polynomials $\bm{q}=(q_1,\dots,q_\ell)$ in
$\KK[\bZ_1,\dots,\bZ_r]$. We say that $\bm{q}$ is $\Sc_{\lambda}$-{\it
  equivariant} if for any $\sigma$ in $\Sc_{\lambda}$ and $i$ in
$1,\dots,\ell$, we have $\sigma(q_i) = q_{\sigma(i)}$, or equivalently
\[
q(z_{\sigma(1)}, \dots, z_{\sigma(\ell)}) = q_{\sigma(i)}(z_1, \dots, z_\ell);
\] 
here, we are implicitly seeing the elements of $\Sc_\lambda$ as
permutations of $\{1,\dots,\ell\}$, as explained in
Section~\ref{ssec:partitions}.

In geometric terms, the zero-set $V(\bq)\subset \KKbar{}^\ell$ of such
a system is $\Sc_{\lambda}$-invariant, even though the equations
themselves may not be invariant. In what follows, we describe how to
derive equations $\bp=(p_1,\dots,p_\ell)$ that generate the same ideal
as $\bq$ (in a suitable localization of $\KK[\bZ_1,\dots,\bZ_r]$) and
are actually $\Sc_{\lambda}$-invariant. We will need an assumption,
discussed below, that $z_i-z_j$ divides $q_i-q_j$ for all pairwise
distinct indices $i,j$.

\begin{example}\label{ex:symmetrize} 
  Let $n=3$ and $\lambda = (1^2\, 2^1)$ so $r=2$, $\ell_1=2$,
  $\ell_2=1$ and $\ell=3$; we have $\Sc_\lambda=\Sc_2 \times \Sc_1$.
  We take $\bq = (q_1, q_2, q_3)$, where
  \begin{align*}
    q_1 &  = z_2 z_3^2 (z_1 + z_2 + 2z_3) + z_1 z_2 z_3^2,\\
    q_2 &  = z_1 z_3^2 (z_1 + z_2 + 2z_3) + z_1 z_2 z_3^2,\\
    q_3 &   = z_1 z_2 z_3 (z_1 + z_2 + 2z_3) + z_1 z_2 z_3^2.
  \end{align*}
  These polynomials satisfy both the equivariance property
  and the divisibility property. Our procedure will produce the 
  following polynomials:
  \begin{align*}
    p_1 &= (z_1+z_2+ 2z_3)z_3, \\
    p_2 &= (z_1+z_2+2z_3)z_2z_3+ (z_1+z_2+2z_3)z_1z_3, \\
    p_3 &= z_1 z_2 z_3 (z_1 + z_2 + 2z_3) + z_1 z_2 z_3^2.
  \end{align*}
  The polynomials $(p_1, p_2, p_3)$ are symmetric in $(z_1, z_2)$ and
  $(z_3)$, that is, are $\Sc_2 \times \Sc_1$-invariant. They generate
  the same ideal as $(q_1,q_2,q_3)$ in the localization
  $\KK[z_1,z_2,z_3]_{(z_1-z_2)(z_1-z_3)(z_2-z_3)}$.
\end{example}

In order to construct a set of invariant generators we make use of
{\em divided differences} of $\bm{q}=(q_1,\dots,q_\ell)$.  These are
defined as $q_{\{i\}} = q_i$ for $i$ in $\{1,\dots,\ell\}$, and for
each set of $k$ distinct integers $I := \{i_1, \dots, i_k\} \subset
\{1, \dots, \ell \}$, with $k\ge 2$,
\begin{equation}\label{recur_dividediff}
  q_{I} =  \frac{q_{\{i_1, \dots, i_{r-1}, i_{r+1},
      \dots, i_k\}} - q_{\{i_1, \dots, i_{q-1}, i_{q+1}, \dots, i_k\}} }{z_{i_r} - z_{i_q}},
\end{equation}
for any choice of $i_r,i_q$ in $I$, with $i_r\ne i_q$.  Indeed, it is
known (see e.g.~\cite[Theorem 1]{FaugereJules12}) that this defines $q_I$
unambiguously (independently of the choice of $i_r,i_q$). Another 
useful property of divided differences is the following:
\begin{itemize}
\item[(i)] if $z_i-z_j$ divides $q_i-q_j$ for all $1 \leq i < j \leq
  \ell$, then $q_{I}$ is a polynomial for all $I \subset \{1, \ldots,
  \ell\}$.
\end{itemize}
The following proposition then gives our construction of the polynomials
$\bm p$. In what follows, for $i \ge 0$, $\eta_i(y_1,\dots,y_s)$
denotes the degree $i$ elementary symmetric function in variables
$y_1,\dots,y_s$.
\begin{proposition}\label{lemma:vector_symme_divided} 
  Suppose the sequence $\bm{q} = (q_1,\dots,q_\ell)$ in
  $\KK[\bZ_1,\dots,\bZ_r]^\ell$ is $\Sc_{\lambda}$-equivariant and
  satisfies $z_i - z_j$ divides $q_i - q_j$ for $1 \leq i < j \leq
  \ell$.  For $0 \leq k \leq r-1$ and $1 \leq j < \ell_{k+1}$, define
  \begin{align*}
    p_{\tau_{k+1}} &= \sum_{i=\tau_{k}+1}^{\tau_{k+1}} \,
    q_{\{i, \tau_{k+1}+1, \dots, \tau_r\}},\\
    p_{\tau_k + j} &= \sum\limits_{s=1}^{j}
    \, \eta_{j - s}(z_{\tau_{k} +
      s+2}, \dots, z_{\tau_{k+1}}) \big(
    \sum\limits_{i=\tau_k+1}^{\tau_k+s} q_{\{i, \tau_k+s+1, \dots,
        \tau_r\}}\big). 
  \end{align*}
  Then the sequence $$\bp = \big( p_{1}, \dots, p_{\tau_1}, \,
  p_{\tau_1+1}, \dots, p_{\tau_2},\ \dots, p_{\tau_{r-1}+ 1}, \dots,
  p_{\tau_r}\big)$$ is in
  $\KK[\bZ_1,\dots,\bZ_r]^{\Sc_{\lambda}}$. If all $q_i$'s have degree
  at most $d$, then $\deg(p_i) \leq d-\ell+i$ holds for
  $i=1,\dots,\ell$. In particular, if $\ell \geq d+2$, then $p_i = 0$
  for all $i=1, \dots, \ell-d-1$.
\end{proposition}
The degree bound comes by inspection.  We defer the rest of the proof
(which follows by induction) to Appendix~\ref{sec:proof}.
  
\medskip
  
We can also show that $\bm{q}$ can be written as a linear
combination of $\bm{p}$, that is, we can find an $\ell \times \ell$
 matrix polynomial $\mat{U}$ such that $\bm{p} \mat{U} = \bm{q} $.
The construction of $\mat{U}$ proceeds as follows. Let $\mat{M}$ be
the block-diagonal matrix with blocks $\mat{M}_1,\dots,\mat{M}_r$
given by
\[
   \mat{M}_{k+1} = \begin{pmatrix}
     1 &  \eta_1(z_{\tau_k+3}, \dots, z_{\tau_{k+1}})
     &   \eta_{2}(z_{\tau_k+3}, \dots, z_{\tau_{k+1}}) & \cdots &
      \eta_{\ell_{k+1}-2}(z_{\tau_k+3}, \dots, z_{\tau_{k+1}})  &  0  \\
     0 & 1  &  \eta_1(z_{\tau_k+4}, \dots, z_{\tau_{k+1}}) & \cdots &
      \eta_{\ell_{k+1}-3}(z_{\tau_k+4}, \dots, z_{\tau_{k+1}})  & 0  \\
     0 & 0 & 1 & \cdots &
      \eta_{\ell_{k+1}-4}(z_{\tau_k+5}, \dots, z_{\tau_{k+1}})  & 0  \\
     \vspace{0.2cm}
     \vdots & \vdots & \vdots &  & \vdots & \vdots \\
     \vspace{0.2cm}
     0 & 0 & 0 & \cdots & 1 & 0 \\
     0 & 0 & 0 & \cdots & 0 & 1
   \end{pmatrix},
 \]
for all $0 \leq k \leq r-1$.  Note that $\det(\mat{M}_{k+1})=1$ for
all $k$, hence $\det(\mat{M})=1$.
 
For a non-negative integer $u$, denote by $\mat{I}_u$ the identity
matrix of size $u$ and by $\mat{0}$ a zero matrix. Then for $k=0,
\dots, r-1$ and $j = 1, \dots, \ell_{k+1}$, we define the following
$\tau_r \times \tau_r$ polynomial matrices. Set $\mat{B}_{\tau_0+1} =
\mat{I}_{\tau_r}$, $\mat{C}_{\tau_0+1} = \mat{I}_{\tau_r}$,
$\mat{D}_{\tau_0+j} = \mat{I}_{\tau_r}$, and
 \begin{equation}\label{BE}
\mat{B}_{\tau_k+j} =  \begin{pmatrix}
  \begin{matrix}
 \mat{I}_{\tau_k}
  \end{matrix}
  & \rvline & \mat{0} & \rvline & \mat{0}\\
\hline
  \mat{0}   & \rvline & \mat E_{k, j} & \rvline & \mat{0} \\ 
\hline 
\mat{0} & \rvline &  \mat{0} & \rvline &
  \begin{matrix}
\mat{I}_{\tau_r - \tau_{k+1}}
  \end{matrix}
\end{pmatrix}, \text{~with~}
\mat{E}_{k, j} = 
\begin{pmatrix}
  \begin{matrix}
 \mat{I}_{j-1}
  \end{matrix} & \rvline & 
\begin{matrix} z_{\tau_k+j} - z_{\tau_k + 1} \\ \vdots \\
  z_{\tau_k+j} - z_{\tau_k + j           -1} 
\end{matrix}& \rvline & \mat{0}\\
\hline
\begin{matrix} 0 & \hdots & 0
\end{matrix} & \rvline & -1 & \rvline & \mat{0}\\
\hline 
\mat{0} & \rvline & 0 & \rvline & \begin{matrix}
 \mat{I}_{\ell_{k+1}-j}
  \end{matrix}
\end{pmatrix}, \\
\end{equation}

\begin{equation}\label{CF}
\mat{C}_{\tau_k+j} = 
 \begin{pmatrix}
  \begin{matrix}
 \mat{I}_{\tau_k}
  \end{matrix} & \rvline &   \mat{0} & \rvline & \mat{0} \\
\hline 
\mat{0}  & \rvline &  \mat{F}_{k, j} & \rvline & \mat{0} \\
\hline 
\mat{0} & \rvline &  \mat{0} & \rvline &
  \begin{matrix}
\mat{I}_{\tau_r - \tau_{k+1}}
  \end{matrix}
\end{pmatrix}, \text{~with~}
 \mat{F}_{k, j} = \begin{pmatrix}
  \begin{matrix}
{\rm {\bf diag}}(z_{\tau_k+j} -
z_{\tau_k + t})_{t=1}^{j           -1}
  \end{matrix} & \rvline & 
\begin{matrix} \mat{0}
\end{matrix} & \rvline & \mat{0} \\
\hline
\begin{matrix} \frac{-1}{j} & \hdots &
  \frac{-1}{j} 
\end{matrix} & \rvline & \frac{-1}{j} & \rvline & \mat{0}  \\
\hline 
\mat{0} &  \rvline & 0 & \rvline & \mat{I}_{\ell_{k+1}-j}
\end{pmatrix}, \\ 
\end{equation}

\begin{equation}\label{DG}
\mat{D}_{\tau_k+j} = 
\begin{pmatrix} 
  \begin{matrix}
{\rm {\bf diag}}(z_{\tau_k+j} - z_t)_{t=1}^{\tau_k}
  \end{matrix} & \rvline &   \mat{0} & \rvline & \mat{0} \\
\hline 
\mat{G}_{k, j}  & \rvline &  \mat{I}_{\ell_{k+1}} & \rvline & \mat{0} \\
\hline 
\mat{0} & \rvline &  \mat{0} & \rvline &
  \begin{matrix}
\mat{I}_{\tau_r - \tau_{k+1}}
  \end{matrix}
\end{pmatrix}, \ 
\mat{G}_{k, j} \, : \, j^{th} {\rm \, row \, is \ } (1,
\dots, 1), {\rm \,rest \, zeros}. 
\end{equation}

Then we have:

\begin{proposition}\label{lemma:vector_symme_divided2} 
  Suppose the sequence $\bm{q} = (q_1,\dots,q_\ell)$ in
  $\KK[\bZ_1,\dots,\bZ_r]^\ell$ satisfies the conditions of
  Proposition \ref{lemma:vector_symme_divided}.  Let $\Delta=\prod_{1
    \le i < j \leq \ell} (z_i-z_j)$ be the Vandermonde determinant
  associated with $z_1,\dots,z_\ell$. Then the matrix $\mat{U}$ in
  $\KK[\bZ_1,\dots,\bZ_r]^{\ell \times \ell}$, defined by
  $$\mat{M} \cdot \mat{U} = \left( \prod_{k=0}^{r-1} \, \prod_{j =
    1}^{\ell_{k+1}} \, \mat{B}_{\tau_k+j} \, \mat{C}_{\tau_k+j} \,
  \mat{D}_{\tau_k + j} \right) $$ has determinant a unit in
  $\KK[\bZ_1,\dots,\bZ_r,1/\Delta]$ and satisfies $\bp \mat{U} = \bq$.
  %\todo{In what order do we read the matrix product?}
\end{proposition}
The proof of Proposition \ref{lemma:vector_symme_divided2} follows by
induction and is deferred to Appendix~\ref{sec:proof2}.

\begin{example} 
  Consider again the polynomials $\bq=(q_1,q_2,q_3)$ and
  $\bp=(p_1,p_2,p_3)$ of Example~\ref{ex:symmetrize}.  The matrix
  $\mat{U}$ which relates $\bp$ to $\bq$ is constructed as
  follows. For $k=0$ and $j=1, 2$ let
  \[   
    \mat B_{1} =
    \begin{pmatrix}
      1 & 0 & 0\\
      0 & 1 & 0 \\
      0 & 0 & 1      
    \end{pmatrix}, \quad \quad \quad
    \mat  C_{1} =
       \begin{pmatrix}
         1 & 0 & 0 \\
         0 & 1 & 0 \\
         0 & 0  & 1
      \end{pmatrix} \quad \quad \quad \quad
        \mat D_ 1 = 
        \begin{pmatrix}
          1 & 0 & 0 \\
          0 & 1 & 0 \\
          0 & 0  & 1
        \end{pmatrix},
      \]
      \[
        \mat B_{2} =
    \begin{pmatrix}
       1 & z_2-z_1  & 0 \\
       0 & -1~ & 0 \\
        0 & 0 & 1    
    \end{pmatrix}, \quad
      \mat C_2 =
      \begin{pmatrix}
        z_2 - z_1 & 0  & 0 \\
        - \frac{1}{2}~~ & - \frac{1}{2}~~  & 0 \\
        0 & 0 & 1
        \end{pmatrix}, 
        \quad
         \mat D_2 = 
        \begin{pmatrix}
          1 & 0 & 0 \\
          0 & 1 & 0 \\
          0 & 0  & 1
          \end{pmatrix}
        \]
while for  $k=1$ and $j=1$ we have
          \[
    \mat B_{3} =
    \begin{pmatrix}
      1 & 0 & 0\\
      0 & 1 & 0 \\
      0 & 0 & -1~~      
    \end{pmatrix}, \quad
    \mat C_{3} =
    \begin{pmatrix}
      1 & 0 & 0\\
      0 & 1 & 0 \\
      0 & 0 & -1~~     
    \end{pmatrix},
    \quad
    \mat D_3 =
    \begin{pmatrix}
     z_3 - z_1 & 0 & 0 \\
     0 &  z_3 - z_2 & 0 \\
     1 & 1 & 1
      \end{pmatrix}. 
    \]
    In the case $\lambda = (1^2 \, 2^1)$, 
    \[
      \mat M = \begin{pmatrix}
        1 & 0 & 0\\
        0 & 1 & 0 \\
        0 & 0 & 1
        \end{pmatrix} 
        \]
        and hence
    \begin{align*}
      \mat U = ( \mat B_1 \mat C_1  \mat D_1) (  \mat B_2 \mat C_2
      \mat D_2 ) ( \mat B_3 \mat C_3  \mat D_3 )
             = \begin{pmatrix}
                 \frac{1}{2}(z_3 - z_1)(z_2 - z_1) & \frac{-1}{~2}(z_2 - z_1)(z_3-z_2) & 0 \\
                 \frac{1}{2}(z_3-z_1) & \frac{1}{2}(z_3-z_2) & 0 \\
                 1 & 1 & 1              
               \end{pmatrix}.
    \end{align*}
    Note that $\det(\mat U) = \frac{1}{2}(z_3-z_1)(z_3 - z_2)(z_2 - z_1)$.
   \end{example}
The formulas defining $\bp$ are straightforward to implement.  The
following proposition describes the resulting algorithm, called {\sf
  Symmetrize}, and gives the cost of this procedure.
\begin{proposition}\label{prop:symmetrize}
  There exists an algorithm ${\sf Symmetrize}(\lambda,\bq)$ which
  takes as input $\bq$ as in
  Proposition~\ref{lemma:vector_symme_divided} and a partition
  $\lambda$ of $ n$, and returns $\bp$ as defined in that
  proposition. For $\bq$ of degree at most $d$, the runtime is
  $\softO(\ell^3 {\ell + d \choose d})$ operations in $\KK$.
\end{proposition}

The proof occupies the rest of this section. Write $\bq=(q_1,\dots,q_\ell)$, and recall the
expressions defining $\bp=(p_1,\dots,p_\ell)$:
for $k=0, \dots,r-1$, we have
$$p_{\tau_{k}+\ell_{k+1}} = \sum_{i=\tau_{k}+1}^{\tau_{k+1}} \,
q_{\{i, \tau_{k+1}+1, \dots, \tau_r\}}$$ and for $j = 1, \dots,
\ell_{k+1}-1$,
\[
p_{\tau_k + j} = \sum\limits_{s=1}^{j}
 \eta_{j - s}(z_{\tau_{k}  + s+2}, \dots, z_{\tau_{k+1}}) \big(
\sum\limits_{i=1}^{s} q_{\{\tau_k+i, \tau_k+s+1, \dots,
    \tau_r\}}\big).
\]
The main issue is to compute the divided differences
$q_{\{\tau_k+i, \tau_k+s+1, \dots,\tau_r\}}$ appearing in these
expressions, for $k=0,\dots,r-1$ and $1 \le i \le s \le
\ell_{k+1}$. Once this is done, the combinations necessary to obtain
$p_{\tau_k + j}$ are easily carried out. The main
ingredient in the proof is the following lemma which describes the
computation of a single divided difference.
\begin{lemma}\label{lem:dd}
  There exists an algorithm ${\sf Divided\_Difference}(\bq,I)$ that
  takes as input $\bq$ as in Proposition~\ref{prop:symmetrize} and a
  subset $I=\{i_1,\dots,i_k\}$ of $\{1,\dots,\ell\}$, and returns
  $q_{I}$.  For $\bq$ of degree at most $d$, the runtime is
  $\softO(\ell {\ell + d \choose d})$ operations in $\KK$.
\end{lemma}
\begin{proof}
  For $j=1,\dots,k-1$, we claim that given $q_{\{i_1,\dots,i_{j-1}\}}$,
  we can obtain $q_{\{i_1,\dots,i_{j}\}}$ using $\softO({\ell + d
    \choose d})$ operations in $\KK$. 

  To see this note that $q_{\{i_1,\dots,i_{k-1}\}}$ has degree at most
    $d$. In order to compute $q_{\{i_1,\dots,i_{j}\}}$, we use
      evaluation / interpolation.  Choosing ${\ell + d \choose d}$
      points as prescribed in~\cite{CaKaYa89}, the algorithm given
      there allows us to compute the values of both numerator and
      denominator in~\eqref{recur_dividediff} in $\softO({\ell + d
        \choose d})$ operations, then compute their ratio, and finally
      interpolate $q_{\{i_1,\dots,i_{j}\}}$ in the same asymptotic
        runtime. The result then follows.
\end{proof}

Our {\sf Symmetrize} algorithm then proceeds as follows. Apply
algorithm {\sf Divided\_Difference} from Lemma \ref{lem:dd} to all
$[\tau_k+i, \tau_k+s+1, \dots,\tau_r]$, for $k=0,\dots,r-1$
and $1 \le i \le s \le \ell_{k+1}$. There are $O(\ell^2)$ such
indices, so this step takes $\softO(\ell^3 {\ell + d \choose d})$
operations in~$\KK$, allowing us to compute all sums
$\sum_{i=1}^{s} q_{\{\tau_k+i, \tau_k+s+1, \dots, \tau_r\}}$ for the
same asymptotic cost.

For $k=0,\dots,r-1$, $j=1,\dots,\ell_{k+1}-1$ and $s=1,\dots,j$,
we then compute the elementary symmetric polynomial $ \eta_{j -
  s}(z_{\tau_{k} + s+2}, \dots, z_{\tau_{k+1}})$, which does not
involve any arithmetic operations. We multiply it by the above sum,
with cost $\softO({\ell + d \choose d})$, since the polynomials
involved in the product have degree sum at most $d$ and at most $\ell$
variables.  Taking all indices $k,j,s$ into account, this adds another
$\softO(\ell^3 {\ell + d \choose d})$ steps to the total.

%%%%%%%%%%%%%%%%%%%%%%%%%%%%%%%%%%%%%%%%%%%%%%%%%%%%%%%%%%%%

\subsection{{\Upsymmrep}s}\label{ssec:dcr}

In this subsection we describe the geometry of $\Sc_n$-orbits in $\KKbar{}^n$,
we define the data structure we will use to represent
$\Sc_n$-invariant sets, and present some basic algorithms related to
it.

\paragraph{The mapping $E_\lambda$ and its fibers.}
For a partition $\lambda= (n_1^{\ell_1} \, n_2^{\ell_2}\, \dots\,
n_r^{\ell_r})$ of $n$, we define the following two subsets of
$\KKbar{}^n$:
\begin{itemize}
\item[(i)]  $\Clambda$ :  the set of all points $\bm\xi$ in $\KKbar{}^n$ that
  can be written as
  \begin{equation} \label{typecrit} 
    {\bm \xi = \big(
  \underbrace{\xi_{1, 1}, \dots,\xi_{1, 1}}_{n_1}, ~~
\dots,~~ \underbrace{\xi_{1, \ell_1}, \dots,\xi_{1, \ell_1}}_{n_1}, ~\dots,~
      \underbrace{\xi_{r, 1}, \dots, \xi_{r, 1}}_{n_r}, ~\dots, ~
      \underbrace{\xi_{r, \ell_r}, \dots, \xi_{r, \ell_r}}_{n_r} \big)}.
\end{equation} 
\item[(ii)]  ${\cal C}^{\rm strict}_\lambda$ :  the set of all  $\bm\xi$ in $\Clambda$
  for which the $\xi_{i, j}$'s in~\eqref{typecrit} are pairwise distinct. 
\end{itemize}

To any point $\bm\xi$ in $\KKbar{}^n$ we can associate its {\em type}:
this is the unique partition $\lambda$ of $n$ such that there exists
$\sigma$ in $\Sc_n$ for which $\sigma(\bm\xi)$ lies in ${\cal C}^{\rm
  strict}_\lambda$. Since all points in an orbit have the same type,
we can then define the type of an orbit as the type of any point in
it.  Any orbit of type $\lambda = (n_1^{\ell_1} \, n_2^{\ell_2}\,
\dots\, n_r^{\ell_r})$ has size 
$$
\gamma_\lambda = {{n}\choose {n_1, \dots,
    n_1, \dots, n_r, \dots, n_r}}= \frac{n!}{ {n_1!}^{\ell_1}
\cdots {n_r!}^{\ell_r} }$$ since the stabilizer of a point in ${\cal
  C}^{\rm strict}_\lambda$ is $\Sc_{n_1}^{\ell_1}\times \cdots \times
\Sc_{n_r}^{\ell_r}$.

Clearly all points in ${\cal C}^{\rm strict}_\lambda$ have type
$\lambda$, but this is not necessarily true for all points in ${\cal
  C}_\lambda$. This can be understood with the help of the refinement
order we introduced in Subsection~\ref{ssec:partitions}, as ${\cal
  C}_\lambda$ contains points of type $\lambda'$ for all $\lambda' \ge
\lambda$.  More precisely, $\Clambda$ is the disjoint union of
all ${\cal C}^{\rm strict}_\lambda$ for all $\lambda' \ge \lambda$.
\begin{example}
  For the  partitions of $n=3$, we have $(1^3) <
(1^1 2^1)< (3^1)$.  In addition,
  \begin{itemize}
  \item[(a)] ${\cal C}_{(1^3)}$ is  $\KKbar{}^3$, while
    ${\cal C}^{\rm strict}_{(1^3)}$ is the set of all points $\bm\xi$
    with pairwise distinct coordinates.
  \item[(b)] ${\cal C}_{(1^1 2^1)}$ is the set of points that can be
    written $\bm\xi=(\xi_{1,1},\xi_{2,1},\xi_{2,1})$, while ${\cal
      C}^{\rm strict}_{(1^1 2^1)}$ is the subset of it where
    $\xi_{1,1} \ne \xi_{2,1}$.
  \item[(c)]   ${\cal C}_{(3^1)}={\cal C}^{\rm strict}_{(3^1)}$ is
    the set of points $\bm\xi$ whose coordinates are all equal.
  \end{itemize}
\end{example}
\medskip
\noindent 
For $\lambda$ as above, we define a mapping $E_\lambda : {\cal
  C}_\lambda \to \KKbar{}^\ell$ by
$$\begin{array}{crcl} E_\lambda: & \bm\xi \text{~as
    in~\eqref{typecrit}} & \mapsto &
  (\eta_i(\xi_{i,1},\dots,\xi_{i,\ell_i}),\dots,\eta_{\ell_i}(\xi_{i,1},\dots,\xi_{i,\ell_i}))_{1
    \le i \le r},
\end{array}
$$ where for $i=1,\dots,r$ and $j=1,\dots,\ell_i$,
$\eta_j(\xi_{i,1},\dots,\xi_{i,\ell_i})$ is the degree $j$ elementary
symmetric function in $\xi_{i,1},\dots,\xi_{i,\ell_i}$.  One should
see this mapping as a means to compress orbits: through the
application of $E_\lambda$, one can represent a whole orbit ${\cal O}$
of type $\lambda$, which has size $\gamma_\lambda$, by the single
point $E_\lambda({\cal O} \cap \Clambda)=E_\lambda({\cal O}
\cap {\cal C}^{\rm strict}_\lambda)$.

To put this into practice, we need to be able to recover an orbit from
its image. Note that the mapping $E_\lambda$ is onto: for
$\bm{\varepsilon}=(\varepsilon_{1,1},\dots,\varepsilon_{r,\ell_r})$ in
$\KKbar{}^\ell$, one can find a point $\bm\xi$ in
the preimage $E_\lambda^{-1}(\bm \varepsilon)$ by finding the roots
$\xi_{i,1},\dots,\xi_{i,\ell_i}$ of
$$P_i(T)=T^{\ell_i} - \varepsilon_{i,1} T^{\ell_i-1} + \cdots +
(-1)^{\ell_i}\varepsilon_{i,\ell_i},$$ for $i=1,\dots,r$. Since we
will use this idea often, we will write
$E_\lambda^*(\bm\varepsilon)=\Sc_n(\bm\xi)$ for the orbit of any such
point $\bm\xi$ in $E_\lambda^{-1}(\bm \varepsilon)$. This is
well-defined, as all points in this fiber are $\Sc_n$-conjugate. More
generally, for a set $G$ in $\KKbar{}^\ell$, we will write
$E_\lambda^*(G)$ for the union of the orbits 
$E_\lambda^*(\bm\varepsilon)$, for $\bm\varepsilon$ in $G$.

The image $E_\lambda(\ClambdaS)$ of those points
having type $\lambda$ is an open subset $O_\lambda
\subsetneq\KKbar{}^\ell$, defined by the conditions that the
polynomials $P_i$ above are pairwise coprime and squarefree. For $\bm
\varepsilon$ in $\KKbar{}^\ell \setminus O_\lambda$, the orbit
$E_\lambda^*(\bm \varepsilon)$ does not have type $\lambda$, but
rather type $\lambda'$, for some partition $\lambda' > \lambda$.
\begin{example}\label{ex:ElG}
  With $n=3$ and $\lambda=(1^1 2^1)$, we have $\ell=2$ and $E_\lambda$
  maps points of the form $(\xi_{1,1},\xi_{2,1},\xi_{2,1})$ to
  $(\xi_{1,1},\xi_{2,1})$. 
  The polynomials $P_1,P_2$ defined in the previous paragraph are
  respectively given by $P_1(T) =T-\varepsilon_{1,1}$ and
  $P_2(T)=T-\varepsilon_{2,1}$, and $O_\lambda$ is defined by
  $\varepsilon_{1,1} \ne \varepsilon_{2,1}$. 

  The point $\bm\varepsilon=(2,3)$ is in $O_\lambda$; the orbit
  $E_\lambda^*(2,3)$ is $\{(2,3,3),(3,2,3),(3,3,2)\}$. On the other
  hand, $\bm\varepsilon=(1,1)$ is not in $O_\lambda$; the orbit
  $E_\lambda^*(1,1)$ is the point $\{(1,1,1)\}$, and it has type
  $(3^1)> (1^1 2^1)$.
  Finally, if we define $G=\{(1,1),(2,3)\}$, then $E_\lambda^*(G)$ is 
  the set $  W=\{(1,1,1),(2,3,3),(3,2,3),(3,3,2)\}.$
\end{example}

We will need an algorithm that computes the type $\lambda'$ of the
orbit $E_\lambda^*(\bm \varepsilon)$, for a given $\bm \varepsilon$ in
$\KK^\ell$, and also computes the value that the actual compression
mapping $E_{\lambda'}$ takes at this orbit. The algorithm's
specification assumes inputs in $\KK$ (since our computation model is
a RAM over $\KK$) but the procedure makes sense over any field
extension of $\KK$. We will use this remark later in the proof of
Lemma~\ref{lemma:decompose}.

\begin{lemma}\label{lem:typeof}
  There exists an algorithm ${\sf Type\_Of\_Fiber}(\lambda, \bm
  \varepsilon)$ which takes as input a partition $\lambda$ of $n$ with
  length $\ell$ and a point $\bm \varepsilon$ in $\KK^\ell$, and
  returns a partition $\lambda'$ of $n$ of length $k$ and a tuple $\bm
  f$ in $\KK^k$, such that
  \begin{itemize}
  \item[(i)] $\lambda'$ is the type of the orbit ${\cal O}:=E_\lambda^*(\bm \varepsilon)$
  \item[(ii)] $E_{\lambda'}({\cal O}\cap {\cal C}^{\rm strict}_{\lambda'})=\{\bm f\}$.
  \end{itemize}
  The algorithm runs in time $\softO(n)$.
\end{lemma}
\begin{proof}
  Write $\bm\varepsilon=(\varepsilon_{1,1},\dots,\varepsilon_{r,\ell_r})$. The points in 
  $E_\lambda^{-1}(\bm \varepsilon)$ are obtained as  permutations of 
  $$    
{\bm \xi = \big(
  \underbrace{\xi_{1, 1}, \dots,\xi_{1, 1}}_{n_1}, ~~
\dots,~~ \underbrace{\xi_{1, \ell_1}, \dots,\xi_{1, \ell_1}}_{n_1}, ~\dots,~
      \underbrace{\xi_{r, 1}, \dots, \xi_{r, 1}}_{n_r}, ~\dots, ~
      \underbrace{\xi_{r, \ell_r}, \dots, \xi_{r, \ell_r}}_{n_r} \big)},
  $$
  where for $i=1,\dots,r$, $\xi_{i,1},\dots,\xi_{i,\ell_i}$ are the roots of 
  $$
  P_i(T) = T^{\ell_i} - \varepsilon_{i,1} T^{\ell_i-1} + \cdots +
  (-1)^{\ell_i} \varepsilon_{i,\ell_i}=0.
  $$ 
  Finding the type of such a point $\bm\xi$ amounts to finding the duplicates among the
  $\xi_{i,j}$'s. Finding such duplicates can be  done by computing the product
  $$
  P = \left ( T^{\ell_1} - \varepsilon_{1,1} T^{\ell_1-1} + \cdots +
  (-1)^{\ell_1} \varepsilon_{1,\ell_1} \right )^{n_1} \cdots \left ( T^{\ell_r}
  - \varepsilon_{r,1} T^{\ell_r-1} + \cdots + (-1)^{\ell_r} \varepsilon_{r,\ell_r} \right
  )^{n_r}
  $$ and its squarefree factorization $P = Q_1^{m_1} \cdots
  Q_s^{m_s},$ with $m_1 < \cdots < m_s$ and all $Q_i$'s squarefree and
  pairwise coprime.  If $k_i=\deg(Q_i)$ then $\bm\xi$ has type
  $\lambda'=(m_1^{k_1}m_2^{k_2} \dots m_s^{k_s})$ with $\lambda' >
  \lambda$. If we write
  $$
  Q_i = T^{k_i} - f_{i,1} T^{k_i-1} + \cdots + (-1)^{k_i} f_{i,k_i},
  \quad 1 \le i \le s,
  $$ 
  then our output is $(\lambda',\bm f)$, where
  $\bm f = (f_{1,1},\dots,f_{s,k_s})$. 

  Using subproduct tree techniques~\cite[Chapter~10]{Gat03} to compute
  $P$ and fast GCD~\cite[Chapter~14]{Gat03}, all computations take
  quasi-linear time $\softO(n)$.
\end{proof}

\begin{example}
  Let $n=3$ and $\lambda=(1^1 2^1)$, with
  $E_\lambda(\xi_{1,1},\xi_{2,1},\xi_{2,1})=(\xi_{1,1},\xi_{2,1})$. We
  saw that for $\bm \varepsilon = (1, 1)$ in $\KK^2$, the orbit
  $E_\lambda^*(1,1)$ is $\{(1, 1, 1)\}$, which has type $\lambda'=(3^1)$.

  Since $n_1=1$ and $n_2=2$, the above algorithm first expands the product
  $(T-1) (T-1)^2$ as $T^3-3T^2+3T-1$, then computes its squarefree
  factorization as $(T-1)^3$. From this, we read off that $s=1$, $m_1=3$
  and $k_1=1$, so that $\lambda'$ is indeed $(3^1)$. The output is
  $(\lambda', E_{\lambda'}(1,1,1))$, the latter being equal to $(1)$.
\end{example}

\paragraph{A data structure for $\Sc_n$-invariant sets.}
The previous setup allows us to represent invariant sets in $\KKbar{}^n$
as follows. Let $W$ be a set in $\KKbar{}^n$, invariant under the
action of $\Sc_n$. For a partition $\lambda$ of $n$ with $\ell$, we write
\begin{equation}\label{eq:defWl}
W_\lambda = \Sc_n(W \cap \ClambdaS) \subset \KKbar{}^n \quad \text{and}\quad
W'_\lambda = E_\lambda(W \cap \ClambdaS) \subset \KKbar{}^\ell,
\end{equation}
where $\Sc_n(W \cap \cal \ClambdaS)$ is the orbit of $W \cap \ClambdaS$
under $\Sc_n$, or, equivalently, the set of points of type $\lambda$
in $W$ (so this matches the notation used in the introduction).

For two distinct partitions $\lambda,\lambda'$ of $n$, $W_\lambda$ and
$W_{\lambda'}$ are disjoint, so that any invariant set $W$ can be
written as the disjoint union $W=\sqcup_{\lambda \vdash n} ~
W_\lambda$.  When $W$ is finite, we then can represent $W_\lambda$ by
describing the image $W'_\lambda$. Indeed, the cardinality of the set
$W'_\lambda$ is smaller than that of the orbit $W_\lambda$ by a factor
of $\gamma_\lambda$, and we can recover $W_\lambda$ as
$W_\lambda=E_\lambda^*(W'_\lambda)$.  Altogether, we are led to the
following definition.
\begin{definition}
  Let $W$ be a finite set in $\KKbar{}^n$, defined over $\KK$ and
  $\Sc_n$-invariant. A {\em {\symmrep}} of $W$ is a sequence
  $(\lambda_i,\mathscr{R}_i)_{1 \le i \le N}$, where the $\lambda_i$'s
  are all the partitions of $n$ for which $W_{\lambda_i}$ is not
  empty, and, for each $i$, $\mathscr{R}_i$ is a zero-dimensional
  parametrization of $W'_{\lambda_i}$.
\end{definition}
\begin{example}\label{ex:dcr}
  Suppose  $n=3$ and
  $$
  W=\{(1,1,1),(2,3,3),(3,2,3),(3,3,2)\}.
  $$ Then with $\lambda=(1^1 2^1)$ we have $W_\lambda
  =\{(2,3,3),(3,2,3),(3,3,2)\}$, $W'_\lambda =\{(2,3)\} \subset
  \KKbar{}^2$ and $\gamma_\lambda=3,$ while with $\lambda'=(3^1)$, we
  have $W_{\lambda'}= \{(1,1,1)\}$, $W'_{\lambda'} =\{(1)\} \subset
  \KKbar{}^1$ and $\gamma_{\lambda'}=1$.

A {\symmrep} of $W$ would consist of $(\lambda,\mathscr{R}_\lambda)$
and $(\lambda',\mathscr{R}_{\lambda'})$, with
$Z(\mathscr{R}_\lambda)=\{(2,3)\}$ and
$Z(\mathscr{R}_{\lambda'})=\{(1)\}$.
\end{example}

Our main algorithm will have to deal with the following situation.  As
input, we will be given a representation of the set $G$ in
$\KKbar{}^\ell$; possibly, some points in $G$ will not be in the open
set $O_\lambda$ (that is, may correspond to orbits having type
$\lambda'$, for some $\lambda'>\lambda$).  As usual, the finite set
$G$ will be described by means of a zero-dimensional parametrization.
Our goal will then be to compute a {\symmrep}
of $E_\lambda^*(G)$.

\begin{example}
  Take $n=3$, and again let $\lambda=(1^1 2^1)$, with
  $E_\lambda(\xi_{1,1},\xi_{2,1},\xi_{2,1})=(\xi_{1,1},\xi_{2,1})$.
  Assume we are given $G=\{(1,1),(2,3)\} \subset\KKbar{}^2$. In this
  case, $E_\lambda^*(G)$ is the set $W$ seen in Examples~\ref{ex:ElG}
  and~\ref{ex:dcr}, and the output we seek is a distinct coordinates
  representation of $W$, as discussed in Example~\ref{ex:dcr}.
\end{example}

\begin{lemma}\label{lemma:decompose}
  There exists a randomized algorithm ${\sf
    Decompose}(\lambda,\mathscr{R})$, which takes as input a partition
  $\lambda$ of $n$ with length $\ell$ and a zero-dimensional
  parametrization $\mathscr{R}$ of a set $G \subset \KKbar{}^\ell$; it
  returns a {\symmrep} of $E_\lambda^*(G)$.
  The expected runtime is $\softO(D^2 n)$ operations in $\KK$, with
  $D=\deg(\mathscr{R})=|G|$.
\end{lemma}
\begin{proof}
In the first step, we apply our algorithm {\sf Type\_Of\_Fiber} from
Lemma \ref{lem:typeof} where the input fiber is given not with
coefficients in $\KK$, but as the points described by $\mathscr{R}$. A
general algorithmic principle, known as {\em dynamic evaluation},
allows us to do this as follows. Let
$\mathscr{R}=((q,v_1,\dots,v_\ell), \mu)$, with $q$ and the $v_i$'s in
$\KK[y]$.  We then call {\sf Type\_Of\_Fiber} with input coordinates
$(v_1,\dots,v_\ell)$, and attempt to run the algorithm over the
residue class ring $\KK[y]/q$, as if $q$ were irreducible. 

If $q$ is irreducible, $\KK[y]/q$ is a field, and we encounter no
problem. However, in general, $\KK[y]/q$ is only a product of fields,
so the algorithm may attempt to invert a zero-divisor. When this
occurs, a ``splitting'' of the computation occurs. This amounts to
discovering a non-trivial factorization of~$q$.  A direct solution
then consists of running the algorithm again modulo the two factors
that were discovered.  Overall, this computes a sequence
$(\mathscr{R}_i,\lambda_i,\bm f_i)_{1\le i \le N}$, where for
$i=1,\dots,N$,
  \begin{itemize}
  \item[(i)] $\mathscr R_i=((q_i,v_{i,1},\dots,v_{i,\ell}), \mu_i)$ is a
    zero-dimensional parametrization that describes a set $F_i \subset F$.
    In addition $F$ is the disjoint union of $F_1,\dots,F_N$;
  \item[(ii)]  $\lambda_i$ is a partitition of $n$, of length $\ell_i$;
  \item[(iii)]  $\bm f_i$ is a sequence of $\ell_i$ elements with
    entries in the residue class ring $\KK[y]/q_i$;  
  \item[(iv)] for any $\bm \varepsilon$ in $F_i$, corresponding to a
    root $\tau$ of $q_i$, ${\sf Type\_Of\_Fiber}(\lambda, \bm
    \varepsilon) = (\lambda_i,\bm f_i(\tau))$.
  \end{itemize}
  Since {\sf Type\_Of\_Fiber} takes time $\softO(n)$, this process 
  takes time $\softO(D^2 n)$, with $D=\deg(\mathscr R)$. The overhead
  $\softO(D^2)$ is the penalty incurred by a straightforward
  application of  dynamic evaluation techniques.
  
  For $i=1,\dots,N$, let $V_i=E_\lambda^{-1}(F_i)$, so that
  $W=\Sc_n(V)$ is the union of the orbits $W_i=\Sc_n(V_i)$. Then, from
  (iv) above we see  that all points in $W_i$ have type 
  $\lambda_i$ and that $(W_i)'_{\lambda_i}$
  is the set $G_i=\{\bm f_i(\tau) \mid q_i(\tau) = 0\} \subset
  \KKbar{}^{\ell_i}$. Using the algorithm
  of~\cite[Proposition~1]{PoSc13b}, we can compute a zero-dimensional
  parametrization $\mathscr{S}_i$ of $G_i$ in time $\softO(D_i^2 n)$,
  with $D_i = \deg(\mathscr{R}_i)$. The total cost is thus $\softO(D^2
  n)$.

  The $\lambda_i$'s may not be pairwise distinct. Up to changing
  indices, we may assume that $\lambda_1,\dots,\lambda_s$ are
  representatives of the pairwise distinct values among them. Then,
  for $i=1,\dots,s$, we compute a zero-dimensional parametrization
  $\mathscr{T}_i$ that describes the union of those
  $Z(\mathscr{S}_j)$, for $j$ such that $\lambda_j=\lambda_i$.  Using
  algorithm~\cite[Lemma~3]{PoSc13b}, this takes a total of $\softO(D^2
  n)$ operations in $\KK$. Finally, we return
  $(\lambda_i,\mathscr{T}_i)_{1 \le i \le s}$.
 \end{proof}

%% \begin{remark}\label{rem:decompose}
%%   Suppose that we apply the previous procedure to
%%   $F=E_\lambda(W_\lambda)$, for an $\Sc_n$-invariant set $W \subset
%%   \KKbar{}^n$. We claim that the output describes the sets
%%   $E_{\lambda'}(W^{\rm strict}_{\lambda'})$, for all partitions
%%   $\lambda' \ge \lambda$ (for which $W^{\rm strict}_{\lambda'}$ is not
%%   empty).  Indeed, since $W_\lambda =
%%   E_\lambda^{-1}(E_\lambda(W_\lambda))$, the procedure computes a
%%   \symmrep of $X=\Sc_n(W_\lambda)$, and one
%%   verifies that for $\lambda' \ge \lambda$ we have $X^{\rm
%%     strict}_{\lambda'}=W^{\rm strict}_{\lambda'}$, whereas $X^{\rm
%%     strict}_{\lambda'}$ is empty for all other partitions $\lambda'$.
%% \end{remark}

%\input{invariants}
%% \input{symmetrize}

\section{Algorithms for computing critical points}
\label{sec:crit}

We can now turn to the main question in this article. Let $\f = (f_1,
\dots, f_s)$ be polynomials in $\KK[x_1, \dots, x_n]^{\Sc_n}$, with $s
\le n$, and with $V=V(\f) \subset \KKbar{}^{n}$ denoting the algebraic set
defined by $f_1=\cdots=f_s=0$. Given a polynomial $\phi$ in $\KK[x_1,
  \dots, x_n]^{\Sc_n}$, we are interested in describing the algebraic
set $W=\crit(\phi,\f)$ defined by the simultaneous vanishing of the polynomials
\begin{equation} \label{critsys}
   f_1, \dots, f_s, \quad M_{s+1}(\jac(\f,\phi)) 
\end{equation}
where $M_{s+1}(\jac(\f, \phi))$ is the set of $(s+1)$-minors of the
Jacobian matrix $\jac(\f, \phi) \in \KK[x_1,\dots,x_n]^{(s+1) \times
  n}$. Equivalently, this is the set of all $\x$ in $V$ at which
$\jac(\f, \phi)$ has rank less than $s+1$.

If we assume that $\jac(\f)$ has full rank $s$ at any point of $V$,
then $V$ is smooth of codimension $s$ (or empty) and $W$ is the set of
critical points of $\phi$ on it. However, most of our discussion can
take place without this assumption. For the sake of simplicity, in any
case, we will still refer to the solutions of~\eqref{critsys} as {\em
  critical points}.

%%%%%%%%%%%%%%%%%%%%%%%%%%%%%%%%%%%%%%%%%%%%%%%%%%%%%%%%%%%%
%%%%%%%%%%%%%%%%%%%%%%%%%%%%%%%%%%%%%%%%%%%%%%%%%%%%%%%%%%
%%%%%%%%%%%%%%%%%%%%%%%%%%%%%%%%%%%%%%%%%%%%%%%%%%%%%%%%%%

\subsection{Description of the algebraic set $W$}

Fundamental to our results is the fact that $W$ is invariant under the
action of the symmetric group. This follows from the next
lemma, being a direct consequence of the chain rule.
\begin{lemma}\label{lemma:eq_partial}
  Let $g$ be in $\KK[x_1, \dots, x_n]$ and $\sigma$ in $\Sc_n$. Then
  for $k$ in $\{1,\dots,n\}$, we have
\begin{equation} \label{partial}
\sigma\left (\frac{\partial g}{\partial x_{k}}\right ) =
\frac{\partial (\sigma(g))}{\partial x_{\sigma(k)}}. 
\end{equation}
\end{lemma}

\begin{corollary}\label{crit_invariant}
 The algebraic set $W$ is $\Sc_n$-invariant.
\end{corollary}
\begin{proof} Let $\bm \xi$ be in $W$ and $\sigma$ be in ${\cal S}_n$. We need to show that $\sigma(\bm \xi)$ is in $W$, that is,
  $f_i(\sigma(\bm \xi)) = 0$ for all $i$ and $\jac(\f, \phi)$ has rank
  at most $s$ at $\sigma(\bm \xi)$.

  The first statement is clear, since $\bm \xi$ cancels $\f$ and $\f$ is
  $\Sc_n$-invariant. For the second claim, since all $f_i$'s and
  $\phi$ are $\Sc_n$-invariant, Lemma \ref{lemma:eq_partial} implies
  that the Jacobian matrix $\jac(\f, \phi)$ at $\sigma(\bm \xi)$ is
  equal to $(\jac(\f,\phi)(\bm \xi)) \bm{A}^{-1}$, where $\bm{A}$ is
  the matrix of $\sigma$.  Therefore, as with $\jac(\f,\phi)(\bm
  \xi)$, it has rank at most $s$ .
\end{proof}
We remark that the proof of the corollary implies a slightly stronger
property, which we already mentioned in the introduction: the system
$f_1, \dots, f_s, M_{s+1}(\jac(\f,\phi)) $ is globally invariant (that
is, applying any $\sigma \in \Sc_n$ permutes these equations, possibly
changing signs). However, our algorithm does not use this fact
directly.

The corollary above also implies that the discussion in
Section~\ref{ssec:dcr} applies to $W$. In particular, for a partition
$\lambda$ of $n$, the sets $W_\lambda$ and $W'_\lambda$
of~\eqref{eq:defWl} are well-defined. In what follows, we fix a
partition $\lambda = (n_1^{\ell_1} \, n_2^{\ell_2}\, \dots\,
n_r^{\ell_r})$ of $n$ and we let $\ell$ be its length; we explain how
to compute a description of $W'_\lambda$ along the lines of
Section~\ref{ssec:dcr}. For this, we let $\bZ_1,\dots,\bZ_r$ be the
indeterminates associated to $\lambda$, as defined in
Section~\ref{ssec:partitions}, with $\bZ_i = z_{i, 1}, \dots, z_{i,
  \ell_i}$. As in that section, we also write all indeterminates
$z_{1,1},\dots,z_{r,\ell_r}$ as $z_1,\dots,z_\ell$.

\begin{definition}\label{def:typecrit}
  With $\lambda$ and $\bZ_1,\dots,\bZ_r$ as above, we define
  $\Tb_\lambda$, the $\KK$-algebra homomorphism $\KK[x_1,\dots,x_n]\to
  \KK[\bZ_1, \dots, \bZ_r]$ mapping $x_1,\dots,x_n$ to
  \begin{equation} \label{typecrit_vars}
    \underbrace{z_{1, 1}, \dots, z_{1, 1}}_{n_1}, ~
    \dots,~ \underbrace{z_{1, \ell_1}, \dots, z_{1, \ell_1}}_{n_1}~
 \dots,~ \underbrace{z_{r, 1}, \dots, z_{r, 1}}_{n_r},~
    \dots,~ \underbrace{z_{r, \ell_r}, \dots, z_{r, \ell_r}}_{n_r}.
  \end{equation}
The operator $\Tb_\lambda$ extends to vectors or matrices of polynomials
entrywise.
\end{definition}

\medskip

We can now define
\begin{equation}\label{jac_lambda}
\f^{[\lambda]} =   \mathbb{T}_\lambda(\f) =  (f_{1}^{[\lambda]},
   \dots,  f_s^{[\lambda]}) \ \text{ and } \  
 \mat{J}^{[\lambda]} =\Tb_\lambda(\jac(\f,
   \phi)) = \big[J^{[\lambda]}_{i,
     j} \big]_{1 \leq i \leq s+1, 1 \leq j \leq n}.
\end{equation}
Notice that for $f$ in $\KK[x_1, \dots, x_n]^{\Sc_n}$, 
and for any indices $j,k$ in $\{1,\dots,n\}$ for 
which $\Tb_\lambda(x_j)=\Tb_\lambda(x_k)$, we have
\begin{equation}\label{eq:equi-partial}
  \mathbb{T}_\lambda\left(\frac{\partial f}{\partial x_j}\right) =
  \mathbb{T}_\lambda\left(\frac{\partial f}{\partial x_k}\right); 
\end{equation}
this follows by applying Lemma \ref{lemma:eq_partial} to $f$ and the
transposition 
$(j\, k)$. Thus
\begin{equation} \label{eq:div_vector}
\Tb_\lambda \left(\frac{\partial f}{\partial x_1}, \dots,
  \frac{\partial f}{\partial x_n} \right) = 
\big (  \underbrace{f_{1, 1}^{[\lambda]}, 
\dots, f_{1, 1}^{[\lambda]}}_{n_1},  \dots, \underbrace{f_{1,
  \ell_1}^{[\lambda]}, \dots, f_{1, \ell_1}^{[\lambda]}}_{n_1},  \dots, 
  \underbrace{f_{r, 1}^{[\lambda]}, 
\dots, f_{ r, 1}^{[\lambda]}}_{n_r},  \dots, \underbrace{f_{r,
  \ell_r}^{[\lambda]}, \dots, f_{r, \ell_r}^{[\lambda]}}_{n_r} \big), 
\end{equation} 
where $f_{i, j}^{[\lambda]}$ are polynomials in the variables $(\bZ_1,\dots,\bZ_r)$.

\begin{lemma} \label{lemma:columns}
  The columns of the transformed Jacobian matrix $\mat J^{[\lambda]}$ have the form:
  \begin{equation} \label{eq:div_jacobian}
    \mat J^{[\lambda]} =
\big (  \underbrace{J_{1, 1}^{[\lambda]}, 
\dots, J_{1, 1}^{[\lambda]}}_{n_1},  \dots, \underbrace{J_{1,
  \ell_1}^{[\lambda]}, \dots, J_{1, \ell_1}^{[\lambda]}}_{n_1},  \dots, 
    \underbrace{J_{r, 1}^{[\lambda]}, \dots, J_{ r, 1}^{[\lambda]}}_{n_r},
    \dots, \underbrace{J_{r, \ell_r}^{[\lambda]}, \dots,
      J_{r, \ell_r}^{[\lambda]}}_{n_r} \big),  
  \end{equation}
\end{lemma}
\begin{proof}
This follows directly from (\ref{eq:div_vector}), since 
  \[
    (J^{[\lambda]}_{s+1, 1}, \dots, J_{s+1,
      n}^{[\lambda]}) = \Tb_\lambda \left(\frac{\partial
        \phi}{\partial x_1}, \dots, \frac{\partial \phi}{\partial x_n}
    \right) \ \text{ and } \ (J^{[\lambda]}_{i, 1}, \dots,
    J_{i, n}^{[\lambda]}) = \Tb_\lambda \left(\frac{\partial
        f_i}{\partial x_1}, \dots, \frac{\partial f_i}{\partial x_n}
    \right) \] for $i=1, \dots, s$, and all polynomials $f_1,\dots,f_s,\phi$ are in 
$\KK[x_1, \dots, x_n]^{\Sc_n}$.
\end{proof}

We will then let $\mat G^{[\lambda]} = [G_{i,j}^{[\lambda]}]_{1 \leq i
  \leq s+1, 1 \leq j \leq \ell}$ be the matrix with entries in
$\KK[\bZ_1, \dots, \bZ_r]$ obtained from $\jac(\f, \phi)$ by first
applying $\Tb_\lambda$ and then keeping only one representative among
all repeated columns highlighted in the previous lemma.

\begin{example}
   Let $s =1$ and $n = 5$, so we consider two polynomials $f_1,\phi$
   in $\KK[x_1,\dots,x_5]$, and take $\lambda = (1^1 \, 2^2)$. Then \[ f_{
     1}^{[\lambda]} (z_{1, 1}, z_{2, 1}, z_{2, 2}) = \Tb_\lambda(f_1) = f_1(z_{1, 1},
   z_{2, 1}, z_{2, 1}, z_{2, 2}, z_{2, 2}),\] and
\[
  \mat G^{[\lambda]} = \begin{pmatrix}
    \Tb_\lambda(\frac{\partial f_1}{\partial x_1}) & \Tb_\lambda(\frac{\partial
      f_1}{\partial x_2}) 
    & \Tb_\lambda(\frac{\partial f_1}{\partial x_4}) \\
   & & \\
   \Tb_\lambda(\frac{\partial \phi}{\partial x_1}) & \Tb_\lambda(\frac{\partial
     \phi}{\partial x_2})
   & \Tb_\lambda(\frac{\partial \phi}{\partial x_4}) 
     \end{pmatrix} \in \KK[z_{1, 1}, z_{2, 1}, z_{2, 2}]^{2 \times 3}. \]
\end{example}

It is easy to see that the polynomials $\f^{[\lambda]}$ are
$\Sc_\lambda$-invariant, where $\Sc_\lambda$ is the permutation
group $\Sc_{\ell_1} \times \cdots \times \Sc_{\ell_r}$ introduced in
the previous section. However, this is  generally not the case for
the entries of $\mat G^{[\lambda]}$.

\begin{lemma} \label{lemma:g}
Let
  $\bm g^{[\lambda]} = (g^{[\lambda]}_{1},\dots,g^{[\lambda]}_{\ell})$
  be a row of $\G^{[\lambda]}$. Then
  \begin{itemize}
  \item[(i)] $z_i - z_j$ divides $g^{[\lambda]}_i - g^{[\lambda]}_j$
    for $1 \leq i < j   \leq \ell$; 
  \item[(ii)] $\bm g^{[\lambda]}$ is $\Sc_\lambda$-equivariant.  
  \end{itemize}
\end{lemma}
\begin{proof} 
  For the sake of definiteness, let us assume that $\bm g^{[\lambda]}$
  is the row corresponding to the gradient of $f_1$, with the other
  cases treated similarly.

  For statement $(i)$, we start from indices $i,j$ as in the lemma and
  let $S$ be the $\KK$-algebra homomorphism $\KK[\bZ_1, \dots, \bZ_r]
  \to \KK[\bZ_1, \dots, \bZ_r]$ that maps $z_i$ to $z_j$, leaving all
  other variables unchanged.
  Let  $u,v$ in $\{1,\dots,n\}$ be  indices such that
  $g_i^{[\lambda]} = \Tb_\lambda(\partial f_1/\partial x_u)$ and
  $g_j^{[\lambda]} = \Tb_\lambda(\partial f_1/\partial x_v)$ and 
  $\sigma \in \Sc_n$  the transposition $(u\,v)$. From
  Lemma~\ref{lemma:eq_partial}, we have that $\sigma(\partial f_1/\partial
  x_u) = \partial f_1/\partial x_v$ and applying $S \circ \Tb_\lambda$ gives
  $S(\Tb_\lambda(\sigma(\partial f_1/\partial x_u))) =
  S(\Tb_\lambda(\partial f_1/\partial x_v)) $.  For any $h \in
  \KK[x_1,\dots,x_n]$ we have, by construction,
  $S(\Tb_\lambda(\sigma(h)))=S(\Tb_\lambda(h))$. Applying this on
  the left-hand side of the previous equality gives
  $S(g_i^{[\lambda]}) = S(g_j^{[\lambda]})$. As a result, $z_i-z_j$
  divides $g_i^{[\lambda]} - g_j^{[\lambda]}$, as claimed.

  For statement $(ii)$, we take indices $k$ in $\{1,\dots,r\}$ and
  $j,j'$ in $\{1,\dots,\ell_k\}$. We let $\sigma \in \Sc_\lambda$ be
  the transposition that maps $(k,j)$ to $(k,j')$ and prove that
  $\sigma(g_{k,j}^{[\lambda]})= g_{k,j'}^{[\lambda]}$. As before,
  there exist indices $u,v$ in $\{1,\dots,n\}$ such that $g_{k,
    j}^{[\lambda]} = \Tb_\lambda(\partial f_1/\partial x_u)$ and
  $g_{k, j'}^{[\lambda]} = \Tb_\lambda(\partial f_1/\partial
  x_v)$. Without loss of generality, assume that $u$ and $v$ are the
  smallest such indices. Then $\Tb_\lambda$ maps
  $x_u,\dots,x_{u+\ell_k-1}$ to $z_{k,j}$ and
  $x_v,\dots,x_{v+\ell_k-1}$ to $z_{k,j'}$.

  Let $\tau \in \Sc_n$ be permutation that permutes
  $(u,\dots,u+\ell_k-1)$ with $(v,\dots,v+\ell_k-1)$. From
  Lemma~\ref{lemma:eq_partial}, we get $\tau(\partial f_1/\partial x_v)
  = \partial f_1/\partial x_u$.  Then
  $\Tb_\lambda(\tau(\partial f_1/\partial x_u)) = \Tb_\lambda(\partial
  f_1/\partial x_v) = g_{k, j'}^{[\lambda]}$.  By construction, the
  left-hand side is equal to $\sigma(\Tb_\lambda(\partial f_1/\partial
  x_u))$, that is, $\sigma(g_{k, j}^{[\lambda]})$. 
\end{proof}

Lemma \ref{lemma:g} implies that we can apply Algorithm {\sf Symmetrize} from
Section~\ref{divided-diffs} to each row of $\mat G^{[\lambda]}$. The result is 
a polynomial matrix $\mat H^{[\lambda]}$ in
$\KK[\bZ_1,\dots,\bZ_r]$,  whose rows are all
$\Sc_\lambda$-equivariant, and such that $\mat H^{[\lambda]} = \mat
G^{[\lambda]} \mat U^{[\lambda]}$, for some polynomial matrix $\mat
U^{[\lambda]}$ in $\KK[\bZ_1,\dots,\bZ_r]^{\ell \times \ell}$.  
Applying Algorithm {\sf Symmetric\_Coordinates} from
Lemma~\ref{elem_slp} to the entries of both $\f^{[\lambda]}$ and $\mat
H^{[\lambda]}$ gives polynomials $\bar \f^{[\lambda]}$ and a
matrix $\bar{ \mat H}^{[\lambda]}$, all with entries in
$\KK[\e_1,\dots,\e_r]$, with variables $\e_i =
e_{i,1},\dots,e_{i,\ell_1}$ for all $i$, and such that $\bm
f^{[\lambda]} = \bar {\bm f}^{[\lambda]}(\b_eta_1,\dots,\b_eta_r)$ and
$\H^{[\lambda]} = \bar {\H}^{[\lambda]}(\b_eta_1,\dots,\b_eta_r)$.

The following summarizes the main properties of this construction. For
the definitions of the sets ${\cal C}_\lambda$, ${\cal C}_\lambda^{\rm
  strict}$, the mapping $E_\lambda$ and the open set $O_\lambda
\subset \KKbar{}^\ell$, see Section~\ref{ssec:dcr}.
\begin{proposition} \label{prop:some_properties} 
  Let $\lambda$  be a partition of $n$ of length $\ell$. 
  \begin{itemize}
  \item[(i)] If $\ell \le s$, then $E_\lambda(W \cap {\cal
    C}_\lambda)$ is the zero-set of $\bar {\bm f}^{[\lambda]}$ in
    $\KKbar{}^\ell$.
  \item[(ii)] If $\ell > s$, then $W'_\lambda=E_\lambda(W \cap {\cal
    C}_\lambda^{\rm strict})$ is the zero-set of $\bar {\bm
    f}^{[\lambda]}$ and all $(s+1)$-minors of $\bar {\H}^{[\lambda]}$
    in $O_\lambda \subset \KKbar{}^\ell$.
  \end{itemize}
\end{proposition}
\begin{proof}
  Let $\bm \xi$ be in the set ${\cal C}_\lambda$ defined in
  Section~\ref{ssec:dcr}, and write 
$$\bm \xi = \big( \underbrace{\xi_{1, 1}, \dots,\xi_{1, 1}}_{n_1}, \dots,
  \underbrace{\xi_{1, \ell_1}, \dots,\xi_{1, \ell_1}}_{n_1}, \dots,
  \underbrace{\xi_{r, 1}, \dots, \xi_{r, 1}}_{n_r}, \dots,
  \underbrace{\xi_{r, \ell_r}, \dots, \xi_{r, \ell_r}}_{n_r} \big).$$
  Set $\bm \zeta =(\xi_{1,1},\xi_{1,2},\dots,\xi_{r,\ell_r})\in
  \KKbar{}^\ell$ and $\bm \varepsilon = E_\lambda(\bm \xi) \in
  \KKbar{}^\ell$. By definition, we have $\f(\bm \xi) = 
  \f^{[\lambda]}(\bm \zeta)$ and $\jac(\f, \phi)(\bm \xi) = \mat
  J^{[\lambda]}(\bm \zeta)$.
  Thus, $\bm \xi$ is in $W \cap {\cal C}_\lambda$ if and only if it
  cancels $\bm f$ and $\jac(\bm f, \phi)$ has rank at most $s$ at $\bm
  \xi$, that is, if $ \f^{[\lambda]}(\bm \zeta)=0$ and $\mat
  J^{[\lambda]}(\bm \zeta)$ has rank at most $s$. The point $\bm \xi$
  is in $W \cap {\cal C}_\lambda^{\rm strict}$ if all the entries of
  $\bm \zeta$ are also pairwise distinct.

   In addition, we have $\f^{[\lambda]}(\bm \zeta)=\bar {\f}^{[\lambda]}(\bm
  \varepsilon)$ and, by construction, ${\rm rank}( \mat J^{[\lambda]}(\bm
  \zeta))={\rm rank}(\mat G^{[\lambda]}(\bm \zeta))$. If $\ell \le s$ then,  since
  $\mat G^{[\lambda]}$ has $\ell$ columns, we see that $\bm \xi$ is in
  $W \cap {\cal C}_\lambda$ if and only if $\bm
  \varepsilon=E_\lambda(\bm \xi)$ cancels
  $\bar{\f}^{[\lambda]}$. Since $E_\lambda: {\cal C}_\lambda \to 
  \KKbar{}^\ell$ is onto, this implies our first claim.

  Suppose further that $\bm \xi$ is in $ {\cal C}_\lambda^{\rm strict}$,
  so that $\bm \varepsilon$ is in $O_\lambda$. From
  Proposition~\ref{lemma:vector_symme_divided}, we have $\mat
  H^{[\lambda]} = \G^{[\lambda]} \U^{[\lambda]}$. Our assumption on
  $\bm \xi$ implies that $\mat U^{[\lambda]}(\bm \zeta)$ is invertible, so
  that $\mat G^{[\lambda]}$ and $\mat H^{[\lambda]}$ have the same
  rank at $\bm \zeta$. Finally, we have $\H^{[\lambda]}(\bm \zeta)=
  \bar {\H}^{[\lambda]}(\bm \varepsilon)$. All this 
  combined shows that $\bm \xi$ is in $W'_\lambda=E_\lambda(W \cap
  {\cal C}^{\rm strict}_\lambda)$ if and only if $\bm
  \varepsilon=E_\lambda(\bm \xi)$ cancels $\bar 
  {\f}^{[\lambda]}$ and all $(s+1)$-minors of $ \bar
  {\H}^{[\lambda]}$.  Since the restriction $E_\lambda: {\cal C}^{\rm 
    strict}_\lambda \to O_\lambda$ is onto, this implies the
  second claim.
\end{proof}

\subsection{The  ${\sf Critical\_Points\_Per\_Orbit}$ algorithm}\label{sec:CPperO}

The main algorithm of this paper is ${\sf
  Critical\_Points\_Per\_Orbit}$ which takes as input symmetric $\f = (f_1,
\ldots, f_s)$ and $\phi$ in $\KK[x_1, \ldots, x_n]$ and, if finite,
outputs a {\symmrep} of the critical point set $W=\crit(\phi, V(\f))$.
Using our notation from Section~\ref{sec:invariant}, this means that
we want to compute zero-dimensional parametrizations of
$W'_\lambda=E_\lambda(W \cap \ClambdaS)$, for all partitions $\lambda$
of $n$ for which this set is not empty. The algorithm is based on
Proposition~\ref{prop:some_properties}, with a minor modification, as
we will see that it is enough to consider partitions of $n$ of length
$\ell$ either exactly equal to $s$, or at least $s+1$.

For any partition $\lambda$, we first 
need to transform $\f$ and $\phi$, in order to obtain the polynomials
in Proposition~\ref{prop:some_properties}.
\begin{lemma} \label{lemma:prepare}
  There exists an algorithm ${\sf Prepare\_F}(\f,\lambda)$ which takes
  as input $\f$ as above and a partition $\lambda$, and returns $\bar
  \f^{[\lambda]}$. If $\f$ has degree at most $d$, the algorithm takes
 $\softO(n {n+d \choose d}{}^2)$ operations in $\KK$.
 Similarly, there exists an algorithm ${\sf Prepare\_F\_H}(\f,\phi,\lambda)$
  which takes as input $\f,\phi$ as above and a partition $\lambda$,
  and returns $\bar \f^{[\lambda]}$ and $\bar {\H}^{[\lambda]}$.
  If $\f$ and $\phi$ have degree at most $d$, then the algorithm takes
  $\softO(n^4 {n+d \choose d}{}^2)$ operations in $\KK$.\end{lemma}
\begin{proof}
  In the first case, applying $\Tb_\lambda$ to $\f$
  takes linear time in the number of monomials $O(n {n+d \choose d})$
  and gives us $\f^{[\lambda]}$.  We then invoke ${\sf
    Symmetric\_Coordinates}(\lambda, \f^{[\lambda]})$, using
  Lemma~\ref{elem_slp}, in order to obtain $\bar \f^{[\lambda]}$  with the
  cost being $\softO(n {n+d \choose d}){}^2$ operations in $\KK$.

  In the second case, we obtain $\f^{[\lambda]}$ as above. We also
  compute the matrix $\jac(\f, \phi)$, which takes $O(n^2 {n+d \choose
    d})$ operations. For the same cost, we apply $\Tb_\lambda$ to all
  its entries and remove redundant columns, as specified in
  Lemma~\ref{lemma:columns}, so as to yield the matrix $
  \G^{[\lambda]}$.  We then apply Algorithm {\sf Symmetrize} from
  Proposition~\ref{prop:symmetrize} to all rows of $
 \G^{[\lambda]}$, which takes $\softO(n^4 {n + d \choose d})$
  operations, and returns $\H^{[\lambda]}$. Finally, we apply
  ${\sf Symmetric\_Coordinates}$ to all entries of this matrix which
  gives $\bar {\H}^{[\lambda]}$ and takes $\softO(n^2 {n + d
    \choose d}{}^2)$ operations in $\KK$.
\end{proof}

At the core of the algorithm, we need a procedure for finding isolated
solutions of certain polynomial systems. In our main algorithm, we
solve such systems using procedures called ${\sf Isolated\_Points}(\bm
g)$ and ${\sf Isolated\_Points}(\bm g, \H, k)$. Given polynomials
$\bm g$, the former returns a zero-dimensional parametrization of the
isolated points of $V(\bm g)$. The latter takes as input polynomials
$\bm g$, a polynomial matrix $\H$ and an integer $k$, and returns a
zero-dimensional parametrization of the isolated points of $V(\bm g,
M_k(\H))$, where $M_k(\H)$ denotes the set of $k$-minors of $\bm
\H$ (note that the former procedure can be seen as a particular case 
of the latter, where we take $\H$ to be a matrix with no row and $k=-1$).
 To establish correctness of the main algorithm, any
implementation of these procedures is suitable. 

Apart from the subroutines discussed above and the function {\sf
  Decompose} from Lemma~\ref{lemma:decompose}, our algorithm also
requires a procedure ${\sf Remove\_Duplicates}(S)$.  This inputs a
list $S=(\lambda_i,\mathscr{R}_i)_{1 \le i \le N}$, where each
$\lambda_i$ is a partition of $n$ and $\mathscr{R}_i$ a
zero-dimensional parametrization.  As all $\lambda_i$'s may not be
distinct in this list, this procedure removes pairs
$(\lambda_i,\mathscr{R}_i)$ from $S$ so as to ensure that all
resulting partitions are pairwise distinct (the choice of which
entries to remove is arbitrary; it does not affect correctness of the
overall algorithm).
	  	 
\begin{algorithm}[h] 	 
  \caption{${\sf Critical\_Points\_Per\_Orbit}(\f, \phi)$}

~\\

  {\bf Input:}  $\f = (f_1, \dots, f_s)$ and $\phi$ in $\KK[x_1,
    \dots, x_n]^{\Sc_n}$ such that $W=\crit(\phi, V(\f))$ is
  finite. \\

  {\bf Output:} A {\symmrep} of $W$.
  \begin{enumerate}
  \item $S = [\ ]$
  \item For $\lambda \vdash n$ of length $s$
    \begin{enumerate}
    \item\label{step:2a} $\bar \f^{[\lambda]} = {\sf Prepare\_F}(\f,\lambda)$
    \item\label{step:2b} $\mathscr{R}_\lambda = {\sf
        Isolated\_Points}(\bar \f^{[\lambda]})$ 
    \item\label{step:2c} append the output of ${\sf Decompose}(\lambda, \mathscr{R}_\lambda)$ to $S$
    \end{enumerate}
  \item For $\lambda \vdash n$ of length in $\{s+1,\dots,n\}$
    \begin{enumerate}
    \item\label{step:3a} $\bar \f^{[\lambda]}, \bar {\H}^{[\lambda]} = {\sf
        Prepare\_F\_H}(\f, \phi,\lambda)$ 
    \item \label{step:3b}$\mathscr{R}_\lambda = {\sf Isolated\_Points}(\bar
      \f^{[\lambda]}, \bar {\H}^{[\lambda]}, s+1)$ 
    \item \label{step:3c}$(\lambda_i,\mathscr{R}_i)_{1\le i \le N} = {\sf
        Decompose}(\lambda, \mathscr{R}_\lambda)$ 
    \item append  $(\lambda_{i_0},\mathscr{R}_{i_0})$ to $S$, where
      $i_0$ is such that 
      $\lambda_{i_0}=\lambda$, if such an $i_0$ exists
    \end{enumerate}
  \item Return ${\sf Remove\_Duplicates}(S)$
  \end{enumerate} 	 
  \label{alg:critperorbit} 	 
\end{algorithm}

\begin{proposition}
  Algorithm {\sf Critical\_Points\_Per\_Orbit} is correct.
\end{proposition}
\begin{proof}
  The goal of the algorithm is to compute zero-dimensional
  representations of $W'_\lambda=E_\lambda(W \cap {\cal C}^{\rm
    strict}_\lambda)$ for all partitions $\lambda$ of $ n$ for which
  this set is not empty.

  To understand the first loop, recall first that $W$ is assumed to be
  finite. Hence this also holds for all $W \cap {\cal C}_\lambda$,
  and thus for all $E_\lambda(W \cap {\cal C}_\lambda)$. As a result,
  for $\lambda$ of length $s$,
  Proposition~\ref{prop:some_properties}(i) implies that at Step~\ref{step:2b}
  , ${\sf Isolated\_Points}(\bar \f_\lambda)$ returns a
  zero-dimensional parametrization of $G:=E_\lambda(W \cap {\cal
    C}_\lambda)$.  Then, we recall from Lemma~\ref{lemma:decompose}
  that the output of ${\sf Decompose}(\lambda, \mathscr{R}_\lambda)$ is a
  {\symmrep} of $E_\lambda^*(G)$. Note that the latter set is the
  orbit of $W \cap \, {\cal C}_\lambda$, that is, the set of all orbits
  contained in $W$ whose type $\lambda'$ satisfies $\lambda' \ge
  \lambda$.
  Taking into account all partitions $\lambda$ of length $s$, the set
  of partitions $\lambda' \ge \lambda$ covers all partitions of length
  $\ell \in\{1,\dots,s\}$, so that at the end of Step 2, we have
  zero-dimensional parametrizations of $W'_\lambda$ for all partitions
  of length $\ell \in\{1,\dots,s\}$ (with possible repetitions). Calling ${\sf
    Remove\_Duplicates}(S)$ will remove any duplicates among this list.

  The second loop deals with partitions $\lambda$ of length at least
  $s+1$.  Since we assume that $W$ is finite, $W'_\lambda$ is finite
  for any such $\lambda$.  Proposition~\ref{prop:some_properties}(ii)
  then implies that the points in $W'_\lambda$ are isolated points of
  the zero-set of $\bar \f^{[\lambda]}$ and of the $(s+1)$-minors of
  $\bar {\H}^{[\lambda]}$. As a result, $W'_\lambda$ is a subset
  of $Z(\mathscr{R}_\lambda)$, for $\mathscr{R}_\lambda$ computed in
  Step~\ref{step:3b} with all other points in $Z(\mathscr{R}_\lambda)$
  corresponding to points in $W$ with type $\lambda' > \lambda$. In
  particular, after the call to {\sf Decompose}, it suffices to keep
  the entry in the list corresponding to the partition $\lambda$, to
  obtain a description of $W'_\lambda$.
\end{proof}

\section{Cost  of the {\sf Critical\_Points\_Per\_Orbit} Algorithm} \label{sec:weightedhomotopy} 

In this section we provide a complexity analysis of our {\sf
  Critical\_Points\_Per\_Orbit} algorithm and hence also complete the
proof of Theorem~\ref{thm:one}.

At the core of the {\sf Critical\_Points\_Per\_Orbit} algorithm is a
procedure, {\sf Isolated\_Points}. Recall that on input polynomials
$\bm g$, a polynomial matrix $\H$ and an integer $k$, it returns a
zero-dimensional parametrization of the isolated points of $V(\bm g,
M_k(\H))$, where $M_k(\H)$ denotes the set of $k$-minors of $
\H$. We apply this procedure to polynomials with entries in
$\KK[\be_1,\dots,\be_r]=\KK[e_{1, 1}, \ldots, e_{1,\ell_1},e_{2, 1},
  \ldots, e_{2,\ell_2}, \dots, e_{r, 1}, \ldots, e_{r,\ell_r}]$.

Rather than using classical methods for solving these polynomial
systems, we use the {\em symbolic homotopy method for weighted domains} given in
\cite{sparse-homotopy}, %\todo{check the title}, 
as this algorithm is the best suited to handle
a weighted-degree structure exhibited by  such systems. Indeed, the
polynomial ring arising from an orbit parameter $\lambda$,
$\KK[\be_1,\dots,\be_r]$, is obtained through a correspondence between
the variable $e_{i,k}$ and the elementary symmetric polynomial
$\eta_{i, k}(x_{j_1}, \dots, x_{j_m})$, for certain indices
$j_1,\dots,j_m$. More precisely, for any $f$ in
$\KK[\bZ_1, \dots, \bZ_r]^{\Sc_\lambda}$, let $\bar f$ be the
polynomial in $\KK[\bm e_1, \dots, \bm e_r]$ satisfying
\begin{equation*}
  f(\bZ_1,\dots,\bZ_r) = \bar f(\bm{\eta}_1,\dots,\bm{\eta}_r),
\end{equation*}
with $\bm\eta_i=(\eta_{i,1},\dots,\eta_{i,\ell_i})$ for all $i$.
Since each $\eta_{i, k}$ has degree $k$, it is natural to assign a
weight $k$ to variable $e_{i,k}$, so that the weighted degree of $\bar
f$ equals the degree of $f$. Our vector of variable weights is then
is $\bm w= (1,\dots,\ell_1,1,\dots,\ell_2,\dots,1,\dots,\ell_r)$.

%%%%%%%%%%%%%%%%%%%%%%%%%%%%%%%%%%%%%%%%%%%%%%%%%%%%%%%%%%%%

\subsection{Solving  weighted determinantal systems}  
  
In this section, we briefly review the algorithm for solving
determinantal systems over a ring of weighted polynomials.

Suppose we work with polynomials in $\KK[\bm Y]=\KK[y_1, \dots, y_m]$,
where each variable $y_i$ has weight $w_i \geq 1$ (denoted by
$\wdeg(y_i)=w_i$). The weighted degree of a monomial
$y_1^{\alpha_1}\cdots y_m^{\alpha_m}$ is then $\sum_{i=1}^m
w_i\alpha_i $, and the weighted degree of a polynomial is the maximum
of the weighted degree of its terms with non-zero coefficients. The
{\em weighted column degrees} of a polynomial matrix is the sequence of
the weighted degrees of its columns, where the weighted degree of a
column is simply the maximum of the weighted degrees of its entries.
    
Let $\EQS = (\eq_1, \dots, \eq_\tau)$ be a sequence of polynomials in
$\KK[\bm Y]$ and $\MAT = [\gmat_{i, j}] \in \KK[\bm Y]^{p \times q} $
a matrix of polynomials such that $p\leq q$ and $m = q-p+\tau+1$, and let
$V_p(\MAT,\EQS)$ denote the set of points in $\KKbar$ at which all
polynomials in $\EQS$ and all $p$-minors of $\MAT$ vanish.  In
\cite{sparse-homotopy}, a symbolic homotopy algorithm for weighted domains
is presented which constructs a symbolic homotopy from a generic start
system to the system defining $V_p(\MAT,\EQS)$ and then uses this to efficiently determine the isolated  points of $V_p(\MAT,\EQS)$.

The main theorem of \cite{sparse-homotopy}, in the special case of
weighted polynomial rings, is given in terms of a number of
parameters.  Let $(\gamma_1, \ldots, \gamma_\tau)$ be the weighted
degrees of $(\eq_1, \ldots, \eq_\tau)$, let $(\delta_1, \ldots,
\delta_q)$ be the weighted column degrees of $\MAT$, let $d$ be the
maximum of the degrees (in the usual sense) of all $\EQS,\MAT$ and set
$$\Gamma = m^2 \binom{m+d}{m} + m^4 {q \choose p}.$$ 

The following quantities are all related to the degrees of some
geometric objects present in the algorithm. We define \[ c =
\frac{\gamma_1 \cdots \gamma_\tau \, \cdot \eta_{m-\tau}(\delta_1, \dots,
  \delta_q)}{w_1\cdots w_m} ~ \text{ and } ~~e = \frac{(\gamma_1+1)
  \cdots (\gamma_\tau+1) \cdot \eta_{m-\tau}(\delta_1+1, \ldots,
  \delta_q+1)}{w_1\cdots w_m}, \] where where $\eta_{n-s}$ is the
elementary symmetric polynomial of degree $n-s$.  For a subset $\bi =
\{i_1, \ldots, i_{m-\tau}\}\subset \{1, \ldots, q\}$, we further let
$(d_{\bi,1},\dots,d_{\bi,m})$ denote the sequence obtained by sorting
$(\gamma_1, \ldots, \gamma_\tau, \delta_{i_1},\dots,
\delta_{i_{m-\tau}})$ in non-decreasing order, and we write
\begin{equation} \label{eq:kappa}
  \kappa_\bi = \max_{1\leq k \leq
    m}(d_{\bi,1} \cdots d_{\bi,k} w_{k+1} \cdots w_m)\quad \text{ and
  }\quad \kappa = \sum_{\bi = \{i_1, \ldots, i_{m-\tau}\}\subset \{1,
    \ldots, q\}} \kappa_\bi.
\end{equation}
Note that without loss of generality, in these equations, we may also
assume that the weights $w_1,\dots,w_m$ are reordered to form a non-decreasing
sequence.

%\todo{in~\cite{sparse-homotopy}, reorder the equations}
\begin{theorem}\cite[Theorem
  5.3]{sparse-homotopy} \label{thm:homotopy} Let $\G$ be a matrix in
  $\KK[\Y]^{p \times q}$ and $\f = (f_1, \dots, f_\tau)$ be polynomials in
  $\KK[\Y]$ such that $p \leq q$ and $m = q-p+\tau+1$. 
  There exists a randomized algorithm which takes $\G$ and $\f$ as
  input and computes a zero-dimensional parametrization of these
  isolated solutions using
  \[\softO\Big( \big(c(e + c^5)+  \, d^2 \,
  \big(\frac{\kappa}{ w_1\cdots w_m}\big)^2\big)m^4 \Gamma\Big)\]
  operations in $\KK$. Moreover, the number of solutions in the output
  is at most $c$. 
\end{theorem}

\medskip 

\noindent When there is no matrix $\MAT$, so $\tau=m$, then the
runtimes reported above remain the same with the term $\Gamma$
becoming $\Gamma = m^2\binom{m+d}{m}$. In this case, the term $\kappa$
is simply equal to $\kappa= \max_{1\leq k \leq m}(\gamma_{1} \cdots
\gamma_{k} w_{k+1} \cdots w_m)$, assuming that the degrees
$\gamma_1,\dots,\gamma_k$ are given in non-decreasing order.

\bigskip

We finish this subsection with an observation in those cases with
$m> q-p+\tau+1$.

\begin{remark} \label{rk:empty}
  Note that when $m > q-p+\tau+1$, then there are no isolated points
  in $V_p(\G, \f)$.  Indeed if we let $I \subset \KKbar[\Y]$ be the
  ideal generated by the $p$-minors of $\G$ then a result due to Eagon
  and Northcott~\cite[Section~6]{Eagon188} implies that all
  irreducible components of $V(I)$ have codimension at most $q-p+1$.
  By Krull's theorem the irreducible components of $V_p(\G, \f) = V(I
  + \langle f_1, \dots, f_\tau \rangle )$ then have codimension at
  most $q-p+1+\tau$. This implies that the irreducible components of
  $V_p(\G, \f)$ in $\KKbar^m$ have dimension at least
  $m-(q-p+\tau+1)$, which is positive when $m > q-p+\tau+1$.
\end{remark}

%%%%%%%%%%%%%%%%%%%%%%%%%%%%%%%%%%%%%%%%%%%%%%%%%%%%%%%%%%%%

\subsection{The complexity of the ${\sf Isolated\_Points}$ procedure}  \label{ssec:complexity}

Estimating the runtimes for the ${\sf Isolated\_Points}$ algorithms
follows from Theorem \ref{thm:homotopy}, for the weighted domains
associated to various partitions of $n$. Thus we let $\lambda =
(n_1^{\ell_1} \, n_2^{\ell_2}\, \dots\, n_r^{\ell_r})$ be a partition
of length $\ell$, with $\ell \geq s$.

The parameters that appear in Theorem~\ref{thm:homotopy} can be
determined as follows. The weights of variables $(\e_1, \dots, \e_r)$
are $\bw = (1, \dots, \ell_1, \dots, 1, \dots, \ell_r)$. For $i=1,
\dots, s$, the weighted degree of $\bar{f}_i^{[\lambda]}$ is the same
as the degree of $f_i^{[\lambda]}$ and so  is at most $d$. 

For $j=1, \dots, \ell$, the weighted column degree of the $j$th column
of $\bar{\mat H}^{[\lambda]}$ is at most $\delta_j=d-1-\ell+j$ (note
that all entries of the Jacobian matrix of $\bm f,\phi$ have degree at
most $d-1$; then apply
Proposition~\ref{lemma:vector_symme_divided}). In particular, if $\ell
> d$, then all entries on the $j$-th column of $\bar{\mat H}^{[\lambda]}$
equal zero for $j=1,\dots,\ell-d$. Finally, in what follows, we
let $$\Gamma= n^2 {{n+d} \choose d}+ n^4{n \choose {s+1}}.$$

\paragraph{Partitions of length $s$.}
\label{subsubsec:equal}
We recall that when the length $\ell$ of the partition $\lambda$
equals $s$, we do not need to deal with a matrix $\bar{\mat
  H}^{[\lambda]}$. In this situation, one only needs to compute the
isolated points of $V(\bar{\f}^{[\lambda]})$.

Consider such a partition $\lambda = (n_1^{\ell_1} \, n_2^{\ell_2}\,
\dots\, n_r^{\ell_r})$ and the corresponding variables $(\e_1, \dots,
\e_r)$, with $\wdeg(e_{i, k}) = k$ for all $i=1,\dots,r$ and
$k=1,\dots,\ell_i$. We make the following claim: {\em if there exists
  $i$ such that $\ell_i > d$, then there is no isolated point in
  $V(\bar{\f}^{[\lambda]})$.}  Indeed, in such a case, variable
$e_{i,\ell_i}$ does not appear in $\bar{\f}^{[\lambda]}$, for weighted
degree reasons, so that the zero-set of this system is invariant with
respect to translations along the $e_{i,\ell_i}$ axis. In particular,
it admits no isolated solution.
  
Therefore we can suppose that all $\ell_i$'s are at most $d$.  In this
case, the quantities $c,e,\kappa$ used in Theorem~\ref{thm:homotopy}
become respectively
$$\mathfrak{c}_\lambda = \frac{d^s}{w_\lambda},\quad
\mathfrak{e}_\lambda = \frac{n(d+1)^s}{w_\lambda},\quad
\mathfrak{\kappa}_\lambda = d^s = w_\lambda \mathfrak{c}_\lambda,$$
with $w_\lambda = \ell_1!\cdots\ell_r!$. In this case Theorem~\ref{thm:homotopy} implies 
that $V(\bar{\f}^{[\lambda]})$ contains at most
$\mathfrak{c}_\lambda$ isolated points, and that and one can compute
all of them using
\[
\softO\left (\big(
\mathfrak{c}_\lambda(\mathfrak{e}_\lambda+\mathfrak{c}_\lambda^5)+  d^2 \mathfrak{c}_\lambda^2\big)n^4 \Gamma_\lambda
\right )  
\subset 
\softO\left ( d^2 \mathfrak{c}_\lambda(\mathfrak{e}_\lambda + \mathfrak{c}_\lambda^5) n^4\Gamma\right )  \] 
operations in $\KK$.

\paragraph{Partitions of length greater than $s$.}\label{subsubsec:greater}
For a partition $\lambda$ of length $\ell$ greater than $s$, we have
to take into account the minors of the matrix $\bar{\H}^{[\lambda]}$.
Note that the assumptions of Theorem~\ref{thm:homotopy} are satisfied:
the matrix $\bar{\H}^{[\lambda]}$ is in $\KK[\e_1, \dots, \e_r]^{(s+1)
  \times \ell}$, with $\ell \ge s+1$, and we have $s$ equations
$\bar{\f}^{[\lambda]}$ in $\KK[\e_1, \dots, \e_r]$, so the number of
variables $\ell$ does indeed satisfy $\ell = \ell-(s+1) + s+ 1$.

We claim that if $\ell > d$, then the algebraic set
$V_{s+1}(\bar{\H}^{[\lambda]}, \bar{\f}^{[\lambda]})$ does not have
any isolated point. Indeed, in this case, we pointed out above that
the columns of indices $1$ to $\ell-d$ in $\bar{\H}^{[\lambda]}$
are identically zero. After discarding these zero-columns from
$\bar{\H}^{[\lambda]}$, we obtain a matrix $\L^{[\lambda]}$ of
dimension $(s+1) \times d$ such that $V_{s+1}(\bar{\H}^{[\lambda]},
\bar{\f}^{[\lambda]}) = V_{s+1}({\L}^{[\lambda]},
\bar{\f}^{[\lambda]})$, and using Remark~\ref{rk:empty} with $p=s+1,
q=d, \tau=s$ and $m \geq \ell$ shows that this algebraic set has no
isolated points.  

Thus, let us now assume that  $\ell \leq d$.
The matrix $\bar{\H}^{[\lambda]}$ has weighted column degrees
$(\delta_1, \dots, \delta_\ell) = (d-\ell, \dots, d-1)$, whereas the
weighted degrees of all polynomials in $\bar{\f}^{[\lambda]}$ is at
most $d$.  To estimate the runtime of
${\sf Isolated\_Points}(\bar{\H}^{[\lambda]}, \bar{\f}^{[\lambda]})$,
we will need the following property.

\begin{lemma} Let $\kappa$ be defined as in~\eqref{eq:kappa} with
  $m=\ell$,
  $\tau=s$, $p=s+1$, $q = \ell$, $(\delta_1, \dots, \delta_\ell) =
  (d-1-\ell, \dots, d-1)$, and $(\gamma_1, \dots, \gamma_s) = (d,
  \dots, d)$. Then, for partitions of length $\ell$ at most $d$,
  one has 
  \[
  \kappa = d^s \eta_{\ell -  s}(d-1, \dots, d-\ell).
  \]
\end{lemma}
\begin{proof}
  Without loss of generality, we reorder the weights $\bw$ as $\bw' =
  (w_1', \dots, w_\ell')$ such that $w_1' \leq \cdots \leq w_\ell'$.

  Take $\bi= (i_1, \dots, i_{\ell-s}) \subset \{1, \dots, \ell\}$, and
  let $d_\bi = (d_{\bi, 1}, \dots, d_{\bi, \ell})$ be the sequence
  obtained by reordering $(d, \dots, d,\delta_{i_1}, \dots,
  \delta_{i_{\ell-s}})$ in non-decreasing order; we first compute the
  value of $\kappa_\bi$ from~\eqref{eq:kappa}. If $d_{\bi, 1}=0$
  (which can happen only if $\ell=d$), then $\kappa_\bi=0$. Otherwise,
  the sequence $d_\bi$ starts with $d_{\bi,1} \ge 1$ and increases
  until index $\ell-s$, after which it keeps the value $d$. On the
  other hand, the ordered sequence of weights never increases by more
  than $1$, so that for all $k=1, \dots, \ell$, we have $w_k' \leq
  d_{\bi, k}$. In this case,
  \[
  \kappa_\bi = \max_{1\leq k \leq \ell}(d_{\bi, 1}\cdots d_{\bi, k}
  w_{k+1} \cdots w_m) = d_{\bi, 1}\cdots d_{\bi, \ell} = d^s
  \delta_{i_1} \cdots \delta_{i_{\ell-s}};
  \]
  note that this equality also holds if $d_{\bi,1}=0$,
  since then both sides vanish.
  Since $\kappa = \sum_{\bi = \{i_1, \ldots, i_{\ell-s}\}\subset
    \{1,\ldots, q\}} \kappa_\bi$, we get 
  \begin{equation} \label{eq:kappa_sp}
    \kappa = \sum_{\bi = \{i_1, \ldots, i_{\ell-s}\}\subset \{1,
      \ldots, \ell\}}  d^s \delta_{i_1} \cdots \delta_{i_{\ell-s}} = d^s
    \eta_{\ell -  s}(d-1, \dots, d-\ell).
      \end{equation} as claimed.
\end{proof}

The procedure ${\sf Isolated\_Points}\big(\bar{\f}^{[\lambda]}, \bar
{\H}^{[\lambda]}\big)$ then uses the algorithm in
Theorem~\ref{thm:homotopy} with input $\big(\bar{\f}^{[\lambda]}, \bar
{\H}^{[\lambda]}\big)$.  
Writing as before $w_\lambda = {\ell_1!  \cdots \ell_r!}$, the
quantities used in the theorem become
\begin{align*}
  \mathfrak{c}_\lambda &=\frac{d^s \eta_{\ell - s}(d-1, \dots, d-\ell)}{w_\lambda},\\
  \mathfrak{e}_\lambda &= \frac{n(d+1)^s\eta_{\ell - s}(d, \dots, d-\ell+1)}{w_\lambda},\\
  \kappa_\lambda&=d^s\eta_{\ell - s}(d-1, \dots, d-\ell) = w_\lambda \mathfrak{c}_\lambda.
\end{align*}
This implies that running ${\sf
  Isolated\_Points}\big(\bar{\f}^{[\lambda]}, \bar
{\H}^{[\lambda]}\big)$ uses
\[
\softO\left ( \big(
\mathfrak{c}_\lambda(\mathfrak{e}_\lambda +
\mathfrak{c}_\lambda^5) + d^2\mathfrak{c}_\lambda^2 
\big) n^4\Gamma
\right )
\]  operations which is again in
\[
\softO\left (
d^2 \mathfrak{c}_\lambda(\mathfrak{e}_\lambda +
\mathfrak{c}_\lambda^5)
n^4\Gamma
\right ).
\]
As before, the number of solutions in the output is at most
$\mathfrak{c}_\lambda$.  

%%%%%%%%%%%%%%%%%%%%%%%%%%%%%%%%%%%%%%%%%%%%%%%%%%%%%%%%%%%%

\subsection{Finishing the proof of Theorem~\ref{thm:one}}

We can now finish estimating the runtime of the {\sf
  Critical\_Points\_Per\_Orbit} Algorithm. For partitions of length
$s$, at Step~\ref{step:2a}, we only need to compute $\bar{\f}^{[\lambda]}$
which takes $\softO(n {n+d \choose d}{}^2)$ operations in $\KK$ as per
Lemma~\ref{lemma:prepare}. At Step~\ref{step:2b}, the procedure ${\sf
  Isolated\_Points}(\bar \f^{[\lambda]})$ takes at most $\softO\left (
d^2 \mathfrak{c}_\lambda(\mathfrak{e}_\lambda +
\mathfrak{c}_\lambda^5) n^4\Gamma \right )$ operations in $\KK$, as we
saw in Subsection~\ref{subsubsec:equal}. The output of this procedure
contains at most $\frak{c}_\lambda$ points; then, by
Lemma~\ref{lemma:decompose}, the cost of the call to {\sf Decompose}
at Step~\ref{step:2c} is $\softO(\mathfrak{c}_\lambda^2 \, n)$, which is
negligible compared to the previous costs.

For partitions of length greater than $s$, computing
$\bar{\f}^{[\lambda]}$ and $\bar {\H}^{[\lambda]}$ at Step~\ref{step:3a} takes
$\softO(n^4 {n+d \choose d}{}^2)$ operations in $\KK$, by
Lemma~\ref{lemma:prepare}. The procedure ${\sf
  Isolated\_Points}\big(\bar{\f}^{[\lambda]}, \bar
{\H}^{[\lambda]}\big)$ at Step~\ref{step:3b} requires at most $\softO\left (
d^2 \mathfrak{c}_\lambda(\mathfrak{e}_\lambda +
\mathfrak{c}_\lambda^5) n^4\Gamma \right )$ operations in $\KK$, as we
saw in Subsection~\ref{subsubsec:greater}. Again, since the
number of solutions in the output is at most $\mathfrak{c}_\lambda$,
the cost of {\sf Decompose} at Step~\ref{step:3c} is still
$\softO(\mathfrak{c}_\lambda^2 \, n)$ which, as before, is negligible
in comparison to the other costs. To complete our analysis, we need the
following lemma.
   \begin{lemma}\label{isolatedbound}
     With all notation being as above, the following holds
     \[
     \sum_{\lambda \vdash n, \ell_\lambda \geq s}\mathfrak{c}_\lambda
     \leq \mathfrak{c} \quad \text{ and }\quad \sum_{\lambda \vdash n,
       \ell_\lambda \geq s}\mathfrak{e}_\lambda \leq \mathfrak{e},
     \] where $\mathfrak{c} = d^s \, {{n+d-1} \choose n}$ and
     $\mathfrak{e} = n \, (d+1)^s \, {{n+d} \choose n}$. 
   \end{lemma}
   \begin{proof} The proof relies on  the
     combinatorics of integer partitions and properties of elementary
     symmetric functions. Details are given in Appendix~\ref{app:bound}.
   \end{proof}

   As a result, the total cost incurred by our calls to ${\sf
     Isolated\_Points}$ and ${\sf Decompose}$ is 
   \[
   \softO\left (
     \mathfrak{c}(\mathfrak{e}+\mathfrak{c}^5)n^9d^2 \left(
       \binom{n+d}{d} + {n \choose {s+1}} \right ) 
   \right ).
   \]
   Since $\binom{n+d}{d} \le (n+1) \binom{n+d-1}{d}$, we will simplify
   this further, by noticing that for $d\ge 2$ we have $\mathfrak{e}=  n \, (d+1)^s \,
   {{n+d} \choose n} \le  n(n+1) d^{5s}  \, {{n+d-1} \choose n}^5 =
   n(n+1) \mathfrak{c}^5$  so this is 
   \[
   \softO\left (
     \mathfrak{c}^6n^{11}d^2 \left( \binom{n+d}{d} + {n \choose {s+1}}
     \right ) 
   \right ).
    \]
    For the remaining operations, the total cost of ${\sf Prepare\_F}$
    and ${\sf Prepare\_F\_H}$ is 
    $$
    n^4 \sum_{\lambda \vdash n, \ell_\lambda \geq s} \binom{n+d}{d}^2.
    $$
    Since $\binom{n+d}{d} \le (n+1) \binom{n+d-1}{d}$, the binomial
    term in the sum is in $O(n^2 \mathfrak{c}^2)$, so the total is
    $O(n^5 \mathfrak{c}^3)$, and can be neglected. Similarly, the cost
    of ${\sf Remove\_Duplicates}$ is negligible. Therefore, the total
    complexity of {\sf Critical\_Points\_Per\_\-Orbit} is then in  
    \[
    \softO\left (
      n^{11}d^{6s+2}  \binom{n+d}{d}^6\left( \binom{n+d}{d} + {n
          \choose {s+1}} \right ) 
    \right )
    \subset
    \left (d^s \binom{n+d}{d}  {n \choose {s+1}} \right)^{O(1)}.
    \]
    Finally, the total number of solutions reported by our algorithm
    is at most $\sum_{\lambda \vdash n, \ell_\lambda \geq
      s}\mathfrak{c}_\lambda$, which itself is at most $
    \mathfrak{c}$.

\section{Experimental results}\label{sec:experiments}

In this section, we report on an implementation and set of
experimental runs supporting the results in this paper. We compare our
${\sf Critical\_Points\_Per\_Orbit}$ algorithm from Section
\ref{sec:CPperO} with a naive algorithm which computes a
zero-dimensional parametrization of $V(I)$, where $I$ is the ideal
generated by $\f$ and the $(s+1)$-minors of $\jac(\f, \phi)$. Since no
implementation of the weighted sparse determinantal homotopy algorithm
is available at the moment, both algorithms use Gr\"obner bases
computations to solve polynomial systems. Furthermore, using
  Gr\"obner bases computations is sufficient to see the advantage of
  our algorithm when the symmetric structure is exploited in our
  algorihm. 

Our experiments are run using the Maple computer algebra system
running on a computer with 16 GB RAM; the Gr\"obner basis computation
in Maple uses the implementation of the $F_4$ and FGLM algorithms from
the FGb package~\cite{FGb}.  The symmetric polynomials $\f$ and $\phi$ are
chosen uniformly at random in $\KK[x_1, \dots, x_n]$, with $\KK = {\rm
  GF}(65521)$, and have the same degree $n$ as the number of
variables, that is, $\deg(f_1) = \cdots = \deg(f_s) = \deg(\phi) =
n$; the number  $s$ of equations $\f$ ranges from $2$ to $n-1$.

Our experimental results support the theoretical advantage gained by
exploiting the symmetric structure of the input polynomials. In
Table~\ref{table:2}, we first report the number of points, denoted by
$D$, that we compute using our algorithm; that is, $D$ is the sum of
the degrees $\deg(\scrR_\lambda)$ that we obtain for all partitions
$\lambda$ of length at least $s$.  The next column is $\big \lceil
\sum_{\ell_\lambda \geq s} \mathfrak{c}_\lambda \big \rceil $, which
is an upper bound on $D$ (here, $\mathfrak{c}_\lambda$ is as
in Subsection~\ref{ssec:complexity}); as we can see, this bound is quite sharp in
general. We next give the upper bound $\mathfrak{c}$
from~\eqref{intro:two}, which we proved in
Lemma~\ref{isolatedbound}. While this bound is sufficient to prove
asymptotic results (for fixed input degree, for instance, see the
discussion in the introduction), we see that it is far from
sharp. 

Finally, we give the number of points $\deg(I)$ computed by the naive
algorithm, together with the upper bound $\tilde{\mathfrak{c}}$
from~\eqref{eq:naive}; in some cases, we did not complete computations
with the naive algorithm, so $\deg(I)$ was unavailable. We see that in 
all cases, the output of our algorithm is significantly smaller than the 
one from the direct approach.

\begin{table}[H]
\begin{center}
    \begin{tabular}{|c|c||c|c|c||c|c|}
  \hline
$n$ &$ s$ & $D$  & $\big \lceil \sum_{\ell_\lambda \geq s} \mathfrak{c}_\lambda \big \rceil  $& $\mathfrak{c}$ & $\deg(I)$ & $\tilde{\mathfrak{c}}$ \\ 
\hline
4 & 2 & 79 & 80 & 560 & 856 & 864  \\
\hline
4 & 3 & 47 & 48 & 2240 & 744 & 768 \\
\hline
5 & 2 & 425 & 432 & 3150 & 15575 & 16000 \\
\hline
5 & 3 & 357 & 370 & 15750 & 18760 & 20000 \\
\hline
5 & 4 & 143 & 157 & 78750 & 11160&  12500 \\
\hline
6 & 2 & 2222  & 2227 & 16632 & - & 337500 \\
\hline
6 & 3 & 2439  & 2453 & 99792 & - & 540000 \\
\hline
6 & 4 & 1482 &  1503 & 598752 & -&  486000 \\
\hline
6 & 5 & 470 & 486 & 3592512 & - & 233280 \\
  \hline
\end{tabular} 
\caption{Degrees and bounds}
\label{table:2}
\end{center}
\end{table}

In Table~\ref{table:1}, we report on our timings in a detailed
fashion. Here, we give the time needed to compute the zero-dimensional
representations $\deg(\scrR_\lambda)$ obtained by our algorithm,
together with their degrees; Time({\sf total}) denotes the total time
spent in our algorithm. On the other hand, Time({\sf naive}) is the
time to compute a zero-dimensional parametrization for the algebraic
set $V(I)$ using the naive algorithm. Experiments are stopped once the
computation has gone past $24$ hours, with the corresponding time
marked with a dash.

In our experiments, the output $\scrR_\lambda$ was always empty for
partitions of length less than $s$. Indeed, for any partition
$\lambda$ of length at most $s-1$, $Z(\scrR_\lambda) =
V(\bar{f}_1^{[\lambda]}, \dots, \bar{f}_s^{[\lambda]})$, where the
$\bar{f}_i^{[\lambda]}$ are $s$ polynomials in less than $s$ variables
derived from the input $\f$. Since the polynomials $\f$ are chosen at
random, the evaluated block symmetric polynomials $f_1^{[\lambda]},
\dots, f_s^{[\lambda]}$ are generic.
\textcolor{black}{ Using~\cite[Proposition 2.1.(ii)]{sparse-homotopy} or
modifiying slightly the proof of~\cite[Proposition 4.5]{sparse-homotopy}}, 
we indeed expect $Z(\scrR_\lambda)$ to be empty for such  partitions $\lambda$ 
of length less than $s$. 
%\todo{really?} \todo{check citation}. 
However, we point out 
that this output can be non-trivial in the general, non-generic case.

\begin{table}[htp]
\small
  \begin{center}
 \begin{tabular}{|c|c||c|c|c|c|c||c|c|}
  \hline
$n$ &$ s$ & ~~ Partition($\lambda$)~~ & Time($\scrR_\lambda$)  & deg($\scrR_\lambda$)  & $ \lceil \mathfrak{c}_\lambda \rceil$
   & Time({\sf total}) & Time({\sf naive}) & $\deg(I)$
   \\ 
 \hline
4 & 2   &\begin{tabular}{lr} 
$~\lambda = (1^4)$  \\
$~\lambda = (1^2 \, 2^1 )$  \\
$~\lambda = (2^2)$  \\
$~\lambda = (1^1 3^1)$ 
 \end{tabular} & \begin{tabular}{lr} 
1.524s   \\ 0.684s   \\ 0.200s  \\  0.380s 
 \end{tabular}  & \begin{tabular}{lr} 
$7$  \\ $48$ \\ $8$  \\ $16$
 \end{tabular}  & \begin{tabular}{lr} 
$8$  \\ $48$ \\ $8$  \\ $16$
 \end{tabular} & 3.136s
& 0.905s~~ &  \begin{tabular}{lr} 
$856$
 \end{tabular}  
 \\ \hline
4 &  3 &   \begin{tabular}{lr} 
$\lambda = (1^4)$  \\
$\lambda = (1^2 \, 2^1 )$
 \end{tabular} &  \begin{tabular}{lr} 
2.497s  \\ 0.772s  
 \end{tabular} & \begin{tabular}{lr} 
$15$  \\ $32$
 \end{tabular}& \begin{tabular}{lr} 
$16$  \\ $32$
 \end{tabular} &  4.468s  & 0.577s~~
                                      &  \begin{tabular}{lr}
                                           $744$  \end{tabular}  
\\ \hline 
5 & 2 &  \begin{tabular}{lr} 
$\lambda = (1^5)$ \\
$\lambda = (1^3 \, 2^1 )$ \\
$\lambda = (1^2 \, 3)$ \\
$\lambda = (1^1 \, 2^2)$ \\
$\lambda = (1^1 \, 4^1)$\\
$\lambda = (2^1 \, 3^1)$
 \end{tabular} & 
\begin{tabular}{lr} 
9.236s \\ 6.832s \\ 2.128s \\ 2.816s \\ 0.316s\\ 0.392s
 \end{tabular} & 
 \begin{tabular}{lr} 
$9$ \\ $142$\\ 112 \\  112 \\ 25 \\ 25
\end{tabular}& 
 \begin{tabular}{lr} 
$11$ \\ $146$\\ 113 \\  113 \\ 25 \\ 25
\end{tabular} & 34.944s    &
2143.144s
              & 15575 
\\ \hline 
5 & 3 &  \begin{tabular}{lr} 
$\lambda = (1^5)$ \\
$\lambda = (1^3 \, 2^1 )$ \\
$\lambda = (1^2 \, 3)$ \\
$\lambda = (1^1 \, 2^2)$
 \end{tabular}
& \begin{tabular}{lr} 
18.829s \\ 18.120s \\ 4.607s \\ 5.316s
 \end{tabular}
 & \begin{tabular}{lr} 
31  \\ 202 \\ 62 \\ 62
 \end{tabular}& \begin{tabular}{lr} 
37  \\ 209 \\ 63 \\ 63
 \end{tabular} & 48.019s  & 3423.660s
              & 18760 
\\ \hline 
5 & 4 &  \begin{tabular}{lr} 
$\lambda = (1^5)$ \\
$\lambda = (1^3 \, 2^1 )$ 
 \end{tabular} & \begin{tabular}{lr} 
 17.080s \\ 12.024s  
 \end{tabular}  &
\begin{tabular}{lr} 
44 \\ 99
 \end{tabular} &
\begin{tabular}{lr} 
53 \\ 105
 \end{tabular} & 37.372s  & 969.396s~ & 11160                  
\\ \hline 
6 & 2 &  \begin{tabular}{lr}$\lambda = (1^6)$\\
$\lambda = (1^4 \, 2^1 )$ \\
$\lambda = (1^3 \, 3)$ \\
$\lambda = (1^2 \, 2^2)$ \\
$\lambda = (2^3)$\\
$\lambda = (1^2 \, 4^1)$ \\
$\lambda = (1^1 \, 2^1 \, 3^1)$ \\
$\lambda = (1^1 \, 5^1)$\\
$\lambda = (2^1\, 4^1)$\\
$\lambda = (3^2)$
 \end{tabular}  & 
\begin{tabular}{lr}
44.979s \\ 94.240s \\ 110.615s \\  413.351s \\ 7.241s\\
15.208s \\  92.589s \\ 0.756s\\ 1.072s\\ 0.956s
 \end{tabular}  & 
\begin{tabular}{lr}13 \\  334 \\  426 \\ 639 \\ 72 \\  216\\ 432 \\
  36 \\  36 \\ 18
 \end{tabular} &  \begin{tabular}{lr}14 \\  338 \\  426 \\ 639 \\ 72 \\  216\\ 432 \\
  36 \\  36 \\ 18
 \end{tabular} &  861.888s & - & -~~~
\\ \hline 
6 & 3 &  \begin{tabular}{lr}$\lambda = (1^6)$ \\
$\lambda = (1^4 \, 2^1 )$\\
$\lambda = (1^3 \, 3)$ \\
$\lambda = (1^2 \, 2^2)$ \\
$\lambda = (2^3)$\\
$\lambda = (1^2 \, 4^1)$ \\
$\lambda = (1^1 \, 2^1 \, 3^1)$
 \end{tabular}&
 \begin{tabular}{lr}
 92.881s \\  773.924s \\ 114.064s \\ 495.432s \\
 7.356s\\ 9.236s \\  17.908s
 \end{tabular} & \begin{tabular}{lr}
 63 \\ 756\\  504 \\ 756 \\  36 \\ 108 \\  216
 \end{tabular} & \begin{tabular}{lr}
 68 \\ 765\\  504 \\ 756 \\  36 \\ 108 \\  216
 \end{tabular} & 1658.071s & - & -~~~ 
\\ \hline 
6 & 4 &  \begin{tabular}{lr}$\lambda = (1^6)$ \\
$\lambda = (1^4 \, 2^1 )$ \\
$\lambda = (1^3 \, 3)$ \\
$\lambda = (1^2 \, 2^2)$
 \end{tabular} & 
\begin{tabular}{lr}
98.312s \\ 591.78s \\ 26.196s \\  46.420s
 \end{tabular} & 
\begin{tabular}{lr}
 142 \\ 800 \\ 216 \\  324
 \end{tabular}&\begin{tabular}{lr}
 153 \\ 810 \\ 216 \\  324
 \end{tabular}  & 842.256s  & - & -~~~ 
\\ \hline 
6 & 5 &  \begin{tabular}{lr}$\lambda = (1^6)$ \\
$\lambda = (1^4\, 2^1)$ \end{tabular}&
\begin{tabular}{lr}154.808s \\
121.768s \end{tabular}&
\begin{tabular}{lr}150 \\ 320\end{tabular}& \begin{tabular}{lr}162 \\
                                              324\end{tabular} & 251.752s 
          & - & -~~~  
\\ \hline 
\end{tabular} 
    \end{center}
\caption{Algorithm Timings}
\label{table:1}
\end{table}

\section{Conclusion and topics for future research}

In this paper we have provided a new algorithm for efficiently
describing the critical point set of a function $\phi$ a variety
$V(\f)$ with $\phi$ and the defining functions of the variety all
symmetric. The algorithm takes advantage of the symmetries and lower
bounds for describing the generators of the set of critical points and
as a result is more efficient than previous approaches.

When $\f=(f_1, \ldots, f_s)\subset {\bf R}[x_1, \ldots, x_n]$, with
${\bf R}$ is a real field, then computing the critical points of
polynomial maps restricted to $V(\f)$ finds numerous applications in
computational real algebraic geometry. In particular such computations
provide an effective Morse-theoretic approach to many problems such as
real root finding, quantifier elimination or answering connectivity
queries (see \cite{Basu2006}).  We view the complexity estimates in
this paper as a possible first step towards better algorithms for
studying real algebraic sets defined by ${\Sc_n}$-invariant
polynomials.

For instance, let $d$ be the maximum degree of the entries in
$\f=(f_1,\dots,f_s)$ and assume that $\f$ generates an
$(n-s)$-equidimensional ideal whose associated algebraic set is
smooth.  Then under these assumptions, we observe that the set
$W(\phi_u, V(\f))$ with
$$\phi_u: (x_1, \ldots, x_n) \to (x_1-u)^2+\cdots+(x_n-u)^2$$ and $u
\in {\bf R}$, 
has a non-empty intersection with all connected components of
$V(\f)\cap {\bf   R}^n$. Hence, when $W(\phi_u, \f)$ is finite for a
generic choice of $u$, then one can use our algorithm to decide
whenever $V(\f)\cap {\bf R}^n$ is empty. This is done in time
polynomial in $d^s, \binom{n+d}{d}, \binom{n}{s+1}$.  

In such cases, for $d,s$ fixed, we end up with a runtime which is
polynomial in $n$ as in \cite{Timofte2003, Riener2012,
  Riener2016}. These latter references are restricted to situations
when $d < n$ is fixed. If now, one takes families of systems where
$d=n$ and $s$ is fixed, we obtain a runtime which is polynomial in
$2^n$. This is an exponential speed-up with the best previous possible
alternatives which run in time $2^{O(n\log(n))}$ as in for example
\cite[Chap.  13]{Basu2006} (but note that these algorithms are
designed for general real algebraic sets).

Obtaining an algorithm to decide whether $V(\f)\cap {\bf R}^n$ is
empty in time polynomial in $d^s, \binom{n+d}{d}, \binom{n}{s+1}$,
without assuming that $W(\phi_u, \f)$ is finite for a generic
$u\in {\bf R}$, is still an open problem.

\paragraph*{Acknowledgements.} G. Labahn is supported by the Natural Sciences
and Engineering Research Council of Canada (NSERC), grant number
RGPIN-2020-04276. \'E. Schost is supported by an NSERC Discovery Grant. T.X. Vu
is supported by a labex CalsimLab fellowship/scholarship. The labex CalsimLab,
reference ANR-11-LABX-0037-01, is funded by the program ``Investissements
d'avenir'' of the Agence Nationale de la Recherche, reference
ANR-11-IDEX-0004-02. M. Safey El Din and T.X. Vu are supported by the ANR grants
ANR-18-CE33-0011 \textsc{Sesame}, ANR-19-CE40-0018 \textsc{De Rerum Natura} and
ANR-19-CE48-0015 \textsc{ECARP}, the PGMO grant \textsc{CAMiSAdo} and the
European Union's Horizon 2020 research and innovation programme under the Marie
Sklodowska-Curie grant agreement N. 813211 (\textsc{POEMA}).

%\vspace{-0.5cm}\input{conclusion}

%\bibliographystyle{plain}
%\bibliography{biblio} 
\appendix

\section{Proof of Proposition \ref{lemma:vector_symme_divided}}\label{sec:proof}

The proof of Proposition \ref{lemma:vector_symme_divided} will be done
in stages. We start with some rather straightforward lemmas.

\begin{lemma}\label{lemma:equivariant}
  Consider an $\Sc_\lambda$-equivariant sequence $\bm{q} =
  (q_1,\dots,q_\ell)$ in $\KK[\bZ_1,\dots,\bZ_r]$. Then, for any $I
  \subset \{1,\dots,\ell\}$ and any $\sigma$ in $\Sc_\lambda$,
  we have $\sigma(q_{I})= q_{\sigma(I)}$.
\end{lemma}
\begin{proof}
  By induction on the size of $I$.
\end{proof}

\begin{lemma} \label{smallsum}
  Consider a sequence $\bm{q} = (q_1,\dots,q_\ell)$ in
  $\KK[\bZ_1,\dots,\bZ_r]$, and suppose that
  \begin{itemize}
  \item[(i)] $z_i - z_j$ divides $q_i - q_j$ for $1 \leq i < j   \leq \ell$,
  \item[(ii)]  $\bq$ is $\Sc_\lambda$-equivariant.  
  \end{itemize}
  Then, for $k$ in $\{1,\dots,r\}$ and $s$ in $\{1,\dots,\ell_k\}$,
  the polynomial $\sum_{i=\tau_k+1}^{\tau_k+s} \, q_{\{i, \tau_k+s+1,
      \dots, \ell\}}$ is invariant under any permutation of
  $\{z_{\tau_k+1},\dots,z_{\tau_k+s}\}$.
\end{lemma}
\begin{proof}
For any $\sigma \in \Sc_\lambda$ permuting only
$\{z_{\tau_k+1},\dots,z_{\tau_k+s}\}$, we have, using the previous lemma,
\[
     \sigma\big( \sum_{i=1}^{\tau_k+s} \, q_{\{i, \tau_k+s+1, \dots, \ell\}}\big) =
      \sum_{i=\tau_k+1}^{\tau_k+s} \, \sigma\big( q_{\{i, \tau_k+s+1, \dots, \ell\}}\big) = 
      \sum_{i=\tau_k+1}^{\tau_k+s} \, q_{\{\sigma(i), \tau_k+s+1, \dots, \ell\}}.
\] 
Since $\sigma$ permutes $\{z_{\tau_k+1},\dots,z_{\tau_k+s}\}$ and the
last sum runs over all $i=\tau_k+1, \dots,\tau_k+s$, it equals
$\sum_{i=\tau_k+1}^{\tau_k+s} \, q_{\{i, \tau_k+s+1, \dots, \ell\}}$.
\end{proof}

\medskip \noindent We can now prove the proposition.  The fact that
all entries of $\bp$ are polynomials follows from our first
assumption.  Proving that they are $\Sc_\lambda$-invariant requires
more work, as we have to deal with numerous cases. While most are
straightforward, the last case does involve nontrivial calculations.

  Fix $k \in \{0,\dots,r-1\}$.  We first prove that for $s$ in
  $\{1,\dots,\ell_{k+1}\}$, $i$ in $\{\tau_k +1,\dots,\tau_{k}+s\}$,
  and  $m$ in $\{0,\dots,r-1\}$, with $m \ne k$, then  the term $q_{\{i,
      \tau_{k} + s+1, \dots, \tau_r\}}$ is symmetric in
  $\{z_{\tau_m+1}, \dots, z_{\tau_{m+1}}\}$.
  Indeed, consider a permutation $\sigma \in \Sc_\lambda$ that acts on
  $\{z_{\tau_m+1}, \dots, z_{\tau_{m+1}}\}$ only. By
  Lemma~\ref{lemma:equivariant}, $\sigma(q_{\{i, \tau_{k}+s+1, \dots,
      \tau_r\}})$ is equal to $q_{\{\sigma(i), \sigma(\tau_{k}+s+1),
      \dots, \sigma(\tau_r)\}}$. If $m < k$, then all indices
  $i,\tau_{k}+s+1,\dots,\tau_r$ are left invariant by $\sigma$ while for $m
  > k$, $[\sigma(i), \sigma(\tau_{k}+s+1), \dots, \sigma(\tau_r)]$ is
  a permutation of $[i,\tau_{k}+s+1,\dots,\tau_r]$. In both cases,
  $q_{\{\sigma(i), \sigma(\tau_{k}+s+1), \dots, \sigma(\tau_r)\}}=q_{\{i,
      \tau_{k}+s+1, \dots, \tau_r\}}$, as claimed.

  Consider first the invariance of  $p_{\tau_{k+1}}$.  By
  Lemma~\ref{smallsum}, the sum $\sum_{i=\tau_{k}+1}^{\tau_{k+1}} \,
  q_{\{i, \tau_{k+1}+1, \dots, \tau_r\}}$ is symmetric in
  $\{z_{\tau_k+1}, \dots, z_{\tau_{k+1}}\}$. Next, for $i$ in
  $\{\tau_k +1,\dots,\tau_{k+1}\}$ and $m$ in $\{0,\dots,r-1\}$, with
  $m \ne k$, each term $q_{\{i, \tau_{k+1}+1, \dots, \tau_r\}}$ is
  symmetric in $\{z_{\tau_m+1}, \dots, z_{\tau_{m+1}}\}$, making use of the previous paragraph
   with $s=\ell_{k+1}$.  As a result,
  $p_{\tau_{k+1}}$ is $\Sc_\lambda$-invariant.

  Next, for $j$ in $\{1,\dots, \ell_{k+1}-1\}$ and $\sigma$
  in $\Sc_\lambda$, we prove that $\sigma(p_{\tau_k + j}) =
  p_{\tau_k + j}$.  Assume first that $\sigma$ acts only on
  $\{z_{\tau_m+1},\dots,z_{\tau_{m+1}}\}$, for some $m$ in
  $\{0,\dots,r-1\}$ with $m \ne k$. For $s$ in $\{1, \dots,
  j\}$, the polynomial $\eta_{j - s}(z_{\tau_{k}
    + s+2}, \dots, z_{\tau_{k+1}})$ depends only on $\{z_{\tau_{k}+1},
  \dots, z_{\tau_{k+1}}\}$ and so is $\sigma$-invariant. Using our earlier argument
  we see that for $i$ in $\{\tau_k+1,\dots,\tau_k+s\}$
  the divided difference $q_{\{i, \tau_k+s+1, \dots, \tau_r\}}$ is
  $\sigma$-invariant. As a result, $p_{\tau_k + j}$ itself
  is $\sigma$-invariant.

  It remains to prove that $p_{\tau_k + j}$ is
  $\sigma$-invariant for a permutation $\sigma$ of $\{{\tau_k+1},
  \dots, {\tau_{k+1}}\}$. We do this first for $\sigma =
  (\tau_k+1, \tau_k+2)$, by proving that all summands in the
  definition of $p_{\tau_k + j}$ are
  $\sigma$-invariant. For any $s$ in $\{2,\dots,
 j\}$, $\eta_{j-s} (z_{\tau_k+s+2}, \dots,
  z_{\tau_k+1})$ does not depend on $(z_{\tau_k+1}, z_{\tau_k+2})$, so
  it is $\sigma$-invariant.  For $s$ in $\{2,\dots, j\}$,
  the sum $\sum_{i=\tau_k+1}^{\tau_k+s} q_{\{i, \tau_k+s+1, \dots,
      \tau_r\}}$ is symmetric in $({\tau_k+1}, {\tau_k+2})$, since
  $\sigma$ just permutes two terms in the sum while for $s = 1$,
  $q_{\{\tau_k+1, \tau_k+2, \dots, \tau_r\}}$ is symmetric in
  $(z_{\tau_k+1}, z_{\tau_k+2})$ by
  Lemma~\ref{lemma:equivariant}. Thus, our claim is proved for $\sigma =
  (\tau_k+1, \tau_k+2)$.  
  
  It remains to prove that  $p_{\tau_k + j}$ is
  invariant in $(z_{\tau_k+2}, \dots, z_{\tau_{k+1}})$.
  For any $t = 1, \ldots, j$, set
  \begin{equation}
    {p_{\tau_k + j, t} = \sum_{s=t}^{j} \,
       \, \eta_{j - s
      }(z_{\tau_k+t+2}, \dots, z_{\tau_{k+1}}) \, \big(
      \sum_{i=\tau_k+1}^{\tau_k+s} \, q_{\{i, \tau_k+s+1, \dots,
          \tau_r\}}\big)}. 
  \end{equation} 
  Then $p_{\tau_k +j} = p_{\tau_k + j, 1}$ and we
  have the recursive identity 
  \begin{equation} \label{recursive_eq} 
  p_{\tau_k + j, t-1} = p_{\tau_k + j, t} +
   \eta_{j - t+1}(z_{\tau_k+t+1},
  \dots, z_{\tau_{k+1}}) \, \big(\sum_{i=\tau_k+1}^{\tau_k+t-1} q_{\{i,
    \tau_k+t, \dots, \tau_{r}\}} \big). 
\end{equation}
  For any $t$, set $\bz_{:t} = (z_{\tau_{k}+1}, \dots, z_{\tau_k+t})$
  and $\bz_{t:} = (z_{\tau_k+t}, \dots, z_{\tau_{k+1}})$. We will show
  that for $t = 1, \ldots, j$, the polynomial $p_{\tau_k +
    j, t}$ satisfies:
  \begin{equation}\label{claim_blocks}
    p_{\tau_k + j, t} \mbox{ is block symmetric in } \bz_{:t}
    \mbox{ and } \bz_{t+1:} 
  \end{equation}
  Taking $t=1$ implies that $p_{\tau_k + j} =
  p_{\tau_k + j, 1}$ is symmetric in $\bz_{2:} =
  (z_{\tau_k+2}, \dots, z_{\tau_{k+1}})$, as claimed.

  To prove statement (\ref{claim_blocks}) we use  decreasing induction on
  $t = j, \dots, 1$. The statement is true when
  $t= j$ since in this case $p_{\tau_k+j,
    j} = \sum_{i=\tau_k+1}^{\tau_k+ j} \, q_{\{i,
      \tau_k+j+1, \dots, \tau_r\}}$, which is symmetric in
  $\bz_{:j}$ by Lemma~\ref{smallsum}, while each summand $q_{\{i,
      \tau_k+j+1, \dots, \tau_r\}}$ is symmetric in
    $\bz_{{j+1}:}$ by Lemma~\ref{lemma:equivariant}.  
  Assume now that (\ref{claim_blocks}) is true for some index $t$ in
  $\{2,\dots, j\}$; we show that it also holds for $t-1$.
  That is, we have $p_{\tau_k+ j, t}$ is block symmetric in
  $\bz_{:t}$ and $\bz_{t+1:}$ and need to show that
  $p_{\tau_k+ j, t-1}$ is block symmetric in $\bz_{:t-1}$
  and $\bz_{t:}$.

  From Lemma~\ref{smallsum}, we have that
  $\sum_{i=\tau_k+1}^{\tau_k+t-1} q_{\{i, \tau_k+t, \dots, \tau_{r}\}}$
  is symmetric in $\bz_{:t-1}$. Furthermore, from our induction
  hypothesis, the polynomial $p_{\tau_k + j, t}$ is
  symmetric in $\bz_{:t-1}$, while $\eta_{j -
    t+1}(z_{\tau_k+t+1}, \dots, \tau_{k+1})$ depends only on
  $\bz_{t:}$. Thus, in view of~\eqref{recursive_eq}, we see that
  $p_{\tau_k +j, t-1}$ is symmetric in $\bz_{:t-1}$. It
  remains to prove that it is also symmetric in $\bz_{t:}$.

  We will prove this by showing $\sigma(p_{\tau_k + j,
    t-1}) = p_{\tau_k + j, t-1}$ for any $\sigma =
  (\tau_k+t+1, \tau_k+\epsilon)$ with $\epsilon \in \{t, t+2, \dots,
  \ell_{k+1}\}$. For any such $\sigma$ with $t+2 \leq \epsilon \leq
  \ell_{k+1}$, our induction hypothesis implies that
  $\sigma(p_{\tau_k+ j, t}) = p_{\tau_k+ j, t}$,
  while $\sigma(\eta_{j - t+1}(z_{\tau_k+t+1}, \dots,
  \tau_{k+1})) = \eta_{j - t+1}(z_{\tau_k+t+1}, \dots,
  \tau_{k+1})$ and $\sigma\big( q_{\{i, \tau_k+t, \dots, \tau_r\}}\big)
  = q_{\{i, \tau_k+t, \dots, \tau_r\}}$ hold for all $i$.  Together with
  (\ref{recursive_eq}), we get 
  $\sigma(p_{\tau_k+ j, t-1}) = p_{\tau_k+j,
    t-1}$.
  Finally, if $\sigma = (\tau_k+t+1, \tau_k+t)$, then we have
  $$\sigma(\eta_{j - t+1}(z_{\tau_k+t+1}, \dots,
  \tau_{k+1})) = \eta_{j - t+1}(z_{\tau_k+t},
  z_{\tau_k+t+2}, \dots, \tau_{k+1})$$ and $\sigma\big(q_{\{i, \tau_k+t,
      \dots, \tau_{r}\}} \big) = q_{\{i, \tau_k+t, \dots, \tau_{r}\}}$
  for all $i=\tau_k+1, \dots, \tau_k+t-1$.  Notice 
  that 
$$ 
\eta_{j - t+1}(z_{\tau_k+t}, z_{\tau_k+t+2}, \dots,
  \tau_{k+1}) - \eta_{j - t+1}(z_{\tau_k+t+1}, \dots,
  \tau_{k+1}) = 
(z_{\tau_k+t} - z_{\tau_k+t+1} ) \,
  \eta_{j-t}(z_{\tau_k+t+2}, 
\dots, z_{\tau_{k+1}}).
$$
Therefore,
 \begin{align} \label{zzz1} \sigma(p_{\tau_k+j, t-1}) -
   p_{\tau_k+\hat{\imath}, t-1}&=  \sigma(p_{\tau_k+j, t})  -
                                 p_{\tau_k+j, t}  \nonumber  \\ 
&\quad + {(z_{\tau_k+t} - z_{\tau_k+t+1} )} \,
  \eta_{j-t}(z_{\tau_k+t+2}, 
\dots, z_{\tau_{k+1}}) \big(\sum_{i = \tau_k+1}^{\tau_k+t-1} \, q_{\{i,
  \tau_k+t, \dots, \tau_{r}\}}  \big) \nonumber  \\
   &= \sigma(p_{\tau_k+j, t}) 
     - p_{\tau_k+j , t} +
     \eta_{j-t}(z_{\tau_k+t+2}, 
\dots, z_{\tau_{k+1}}) \nonumber  \\
   & \quad    \big(\sum_{i = \tau_k+1}^{\tau_k+t-1} \, ( q_{\{i,
     \tau_k+t+1, \tau_k+t+2,  \dots, \tau_{r}\}} - q_{\{i, \tau_k+t,
     \tau_k+t+2,  \dots, \tau_{r}\}} )\big),
\end{align}
where the last equality follows from the definition of
divided differences. In
  particular, 
  $$
  \sigma(p_{\tau_k+j, j-1}) -
  p_{\tau_k+j, j-1} = 
\sigma(p_{\tau_k+j, j})   
     - p_{\tau_k+j, j} +
     \sum_{i=\tau_k+1}^{\tau_k+j-1}
     (q_{\{i, \tau_k+j+1, \dots, \tau_r\}} - q_{\{i, \tau_k+j,
       \tau_k+j+2, \dots, \tau_r\}}).
       $$ 
      In addition, since
     $p_{\tau_k+j, j} = \sum_{i=\tau_k+1}^{\tau_k+j}
     q_{\{i, \tau_k+j+1, \dots, \tau_r\}}$, then when $\sigma = (\tau_k
     + j+1, \tau_k+j)$, we have $\sigma(p_{\tau_k+j,
       j}) - p_{\tau_k+j,j} =
     \sum_{i=\tau_k+1}^{\tau_k+j-1}(q_{\{i, \tau_k+j,\tau_k+j+2,
       \dots, \tau_r\}} - q_{\{i, \tau_k+j+1, \dots, \tau_r\}})$. This implies that
     $\sigma(p_{\tau_k+j, j-1}) - p_{\tau_k+j, j-1} = 0$. 

When $t \leq j-1$, from~(\ref{recursive_eq}),
   taken at index $t+1$, {if $\sigma = (\tau_k+t+1, \tau_k+t)$}, we also have
  \begin{align*}
    \sigma(p_{\tau_k+j, t}) 
&= \sigma(p_{\tau_k+j, t+1}) +
     \eta_{j-t}(z_{\tau_k+t+2}, \dots, z_{\tau_{k+1}})
    \big(\sum_{i=\tau_k+1}^{\tau_k+t-1}  q_{\{i,\tau_k+t,
      \tau_k+t+2, \dots, \tau_{r}\}} + q_{\{\tau_k+t, \tau_k+t+1, \dots,
      \tau_{k+1}\}} \big).
  \end{align*}
Then, by subtraction:

    \begin{align*} \sigma(p_{\tau_k+j, t}) - p_{\tau_k+j, t}
      &= \sigma(p_{\tau_k+j, t+1})- p_{\tau_k+j,
        t+1} +  \eta_{j-t}(z_{\tau_k+t+2}, \dots, z_{\tau_{k+1}}) \\
      &~~~ \big( \sum_{i=\tau_k+1}^{\tau_k+t-1} (q_{\{i, \tau_k+t, \tau_k+t+2, \dots, \tau_r\}}
        - q_{\{i, \tau_k+t+1, \dots, \tau_r\}})\big)
      \end{align*}
and so
       \begin{align} \label{zzz}
         \sigma(p_{\tau_k+j, t+1})- p_{\tau_k+j, t+1}
         =& \sigma(p_{\tau_k+j, t}) - p_{\tau_k+j, t} + 
           \eta_{j-t}(z_{\tau_k+t+2}, \dots, z_{\tau_{k+1}}) \nonumber  \\
         & ~~~  \big( \sum_{i=\tau_k+1}^{\tau_k+t-1} (q_{\{i, \tau_k+t+1, \dots, \tau_r\}}
        - q_{\{i, \tau_k+t, \tau_k+t+2, \dots, \tau_r\}})\big).
      \end{align}
Combining~(\ref{zzz1}) and~(\ref{zzz}) gives
$\sigma(p_{\tau_k+j, t-1}) - p_{\tau_k+j, t-1} =
\sigma(p_{\tau_k+j, t+1}) - p_{\tau_k+j, t+1}$.
By induction, we have that $p_{\tau_k+j, t+1}$ is symmetric
in $\bz_{:t+1}$ and so $\sigma(p_{\tau_k+j, t+1}) =
p_{\tau_k+j, t+1}$ for $\sigma = (\tau_k+t+1, \tau_k+t)$
which in turn implies that $\sigma(p_{\tau_k+j, t-1}) =
p_{\tau_k+j, t-1}$.  This gives our result.

\section{Proof of Proposition \ref{lemma:vector_symme_divided2}}\label{sec:proof2}

%\begin{proof}[Constructing the $\mat U$ matrix.]
Define the row vector 
$$
\h = \big(h_{\tau_0+1}, \dots,h_{\tau_1}, \dots, h_{\tau_{r-1}+1}, \dots,h_{\tau_r}\big)
$$ 
where, for $k=0, \dots, r-1$ and $j = 1, \dots, \ell_{k+1}$,
\begin{align}\label{eqdef:hs}
h_{\tau_k+j} = \sum_{i =
  \tau_k+1}^{\tau_k+j} q_{\{i, \tau_k+j+1,
  \dots, \tau_{r}\}}. 
\end{align}
  
Then for all $i=1, \dots, m$, $k=0, \dots, r-1$, $
   p_{\tau_k+\ell_{k+1}} =   h_{\tau_k+\ell_{k+1}}$, and for $j = 1, \dots, \ell_{k+1}-1$, 
 \[ p_{\tau_k+j} = \sum_{s=1}^{j} \, 
   \eta_{j-s}(z_{\tau_k+s+2}, \dots, z_{\tau_{k+1}}) \, h_{\tau_k+s}.
 \]
Then $\h = \bp \, \mat{M}$, where we recall that $\mat{M}$ is
the block-diagonal matrix with blocks 
$\mat{M}_1,\dots,\mat{M}_r$ where
 \[
   \mat{M}_{k+1} = \begin{pmatrix}
     1 &  \eta_1(z_{\tau_k+3}, \dots, z_{\tau_{k+1}})
     &   \eta_{2}(z_{\tau_k+3}, \dots, z_{\tau_{k+1}}) & \cdots &
      \eta_{\ell_{k+1}-2}(z_{\tau_k+3}, \dots, z_{\tau_{k+1}})  &  0  \\
     0 & 1  &  \eta_1(z_{\tau_k+4}, \dots, z_{\tau_{k+1}}) & \cdots &
      \eta_{\ell_{k+1}-3}(z_{\tau_k+4}, \dots, z_{\tau_{k+1}})  & 0  \\
     0 & 0 & 1 & \cdots &
      \eta_{\ell_{k+1}-4}(z_{\tau_k+5}, \dots, z_{\tau_{k+1}})  & 0  \\
     \vspace{0.2cm}
     \vdots & \vdots & \vdots &  & \vdots & \vdots \\
     \vspace{0.2cm}
     0 & 0 & 0 & \cdots & 1 & 0 \\
     0 & 0 & 0 & \cdots & 0 & 1
   \end{pmatrix} .
 \]
 Then $\det(\mat{M})=1$ and  $\mat N = {\mat M}^{-1}$  is also a polynomial
 matrix in $\KK[\bZ]$ with $\det({\mat N}) = 1$. 

 We  construct a matrix $\mat{J}$ which defines the column
 operations converting $\h$ into $\bq$ as follows. 
%  For a non-negative integer
% $u$, denote by $\mat{I}_u$ the identity matrix of size $u$ and by
% $\mat{0}$ a zero matrix. 
Recall that for $k=0, \dots, r-1$ and $j = 1, \dots,
 \ell_{k+1}$, we have defined the following $\tau_r \times \tau_r$
 polynomial matrices. Set $\mat{B}_{\tau_0+1} = \mat{I}_{\tau_r}$,
 $\mat{C}_{\tau_0+1} = \mat{I}_{\tau_r}$, $\mat{D}_{\tau_0+j} =
 \mat{I}_{\tau_r}$, and
\[
\mat{B}_{\tau_k+j} =  \begin{pmatrix}
  \begin{matrix}
 \mat{I}_{\tau_k}
  \end{matrix}
  & \rvline & \mat{0} & \rvline & \mat{0}\\
\hline
  \mat{0}   & \rvline & \mat E_{k, j} & \rvline & \mat{0} \\ 
\hline 
\mat{0} & \rvline &  \mat{0} & \rvline &
  \begin{matrix}
\mat{I}_{\tau_r - \tau_{k+1}}
  \end{matrix}
\end{pmatrix}, \text{~with~}
\mat{E}_{k, j} = 
\begin{pmatrix}
  \begin{matrix}
 \mat{I}_{j-1}
  \end{matrix} & \rvline & 
\begin{matrix} z_{\tau_k+j} - z_{\tau_k + 1} \\ \vdots \\
  z_{\tau_k+j} - z_{\tau_k + j           -1} 
\end{matrix}& \rvline & \mat{0}\\
\hline
\begin{matrix} 0 & \hdots & 0
\end{matrix} & \rvline & -1 & \rvline & \mat{0}\\
\hline 
\mat{0} & \rvline & 0 & \rvline & \begin{matrix}
 \mat{I}_{\ell_{k+1}-j}
  \end{matrix}
\end{pmatrix}; \\
\]

\[
\mat{C}_{\tau_k+j} = 
 \begin{pmatrix}
  \begin{matrix}
 \mat{I}_{\tau_k}
  \end{matrix} & \rvline &   \mat{0} & \rvline & \mat{0} \\
\hline 
\mat{0}  & \rvline &  \mat{F}_{k, j} & \rvline & \mat{0} \\
\hline 
\mat{0} & \rvline &  \mat{0} & \rvline &
  \begin{matrix}
\mat{I}_{\tau_r - \tau_{k+1}}
  \end{matrix}
\end{pmatrix}, \text{~with~}
 \mat{F}_{k, j} = \begin{pmatrix}
  \begin{matrix}
{\rm {\bf diag}}(z_{\tau_k+j} -
z_{\tau_k + t})_{t=1}^{j           -1}
  \end{matrix} & \rvline & 
\begin{matrix} \mat{0}
\end{matrix} & \rvline & \mat{0} \\
\hline
\begin{matrix} \frac{-1}{j} & \hdots &
  \frac{-1}{j} 
\end{matrix} & \rvline & \frac{-1}{j} & \rvline & \mat{0}  \\
\hline 
\mat{0} &  \rvline & 0 & \rvline & \mat{I}_{\ell_{k+1}-j}
\end{pmatrix}; \\ 
\]

\[
\mat{D}_{\tau_k+j} = 
\begin{pmatrix} 
  \begin{matrix}
{\rm {\bf diag}}(z_{\tau_k+j} - z_t)_{t=1}^{\tau_k}
  \end{matrix} & \rvline &   \mat{0} & \rvline & \mat{0} \\
\hline 
\mat{G}_{k, j}  & \rvline &  \mat{I}_{\ell_{k+1}} & \rvline & \mat{0} \\
\hline 
\mat{0} & \rvline &  \mat{0} & \rvline &
  \begin{matrix}
\mat{I}_{\tau_r - \tau_{k+1}}
  \end{matrix}
\end{pmatrix}, \ 
\mat{G}_{k, j} \, : \, j^{th} {\rm \, row \, is \ } (1,
\dots, 1); {\rm \,others \, are \, zeros}. 
\]

Let 
\[
\mat{J} =  \prod_{k=0}^{r-1} \, \prod_{j = 1}^{\ell_{k+1}}
\, \mat{B}_{\tau_k+j} \, \mat{C}_{\tau_k+j} \, \mat{D}_{\tau_k +
  j}  ~~ \in  ~~ \KK[\bZ_1, \dots, \bZ_r]^{\tau_r \times \tau_r}.
\]
We will prove that this matrix satisfies $\bq = \h \, \mat{J}$. Note first 
that, for $k=0, \ldots, r-1$ and $j~=~1,~\ldots,~\ell_{k+1}$  we 
have $\det(\mat{B}_{\tau_k+j}) = \det(\mat{E}_{k, j}) = -1$, 
$\det(\mat{C}_{\tau_k+j}) = \det(\mat{F}_{k, j}) = \frac{-1}{j} \, 
\prod_{t=1}^{j-1} \, (z_{\tau_k+j} - z_t )$, and 
$\det(\mat{D}_{\tau_k+j}) = \prod_{t=1}^{\tau_k} (z_{\tau_k+j} - z_t)$. This implies that
\[
\det(\mat J) = \alpha \,  \prod_{k=0}^{r-1} \, \prod_{j=1}^{\ell_{k+1}} \, 
\prod_{t=1}^{j-1} \, (z_{\tau_k+j} - z_t ) \, \prod_{t=1}^{\tau_k}
(z_{\tau_k+j} - z_t) =\alpha \Delta \text{~for~some~} \alpha \in \KK_{\ne 0}. 
\]  
Define $\mat U = \mat N \,  \mat J$. Then $\bp = \bq \, \mat U$, and $\det(\mat U)$ 
is a unit in $\KK[\bZ_1,\dots,\bZ_r,1/\Delta]$, as claimed.

It remains to prove $\bq = \h \, \mat J$. For $s=0,\dots,\tau_r$,
define
$$\bq_{s} =   \begin{pmatrix}
    q_{\{1, s+1, \dots, \tau_r\}} & \hdots &
    q_{\{s, s+1, \dots, \tau_r\}} &
    h_{s+1} & \hdots & h_{\tau_r} 
    \end{pmatrix},$$
so that for $s=0$ we have $\bq_0 = \h$, whereas for $s=\tau_r$ we 
have $\bq_{\tau_r} = \bq$. We prove the following: for $k$ in $\{0,\dots,r-1\}$
and $j$ in $\{1,\dots,\ell_{k}\}$, 
\begin{equation}\label{eq:recQ}
  \bq_{\tau_k + j} = \bq_{\tau_k + {j-1}} \mat{B}_{\tau_k+j} \,
  \mat{C}_{\tau_k+j} \, \mat{D}_{\tau_k + j}.
\end{equation}
Our claim $\bq = \h \, \mat J$ then follows from
a direct induction, taking into account the values of $\bq_0$ and
$\bq_{\tau_r}$ given above.

Take $k$ in $\{0,\dots,r-1\}$ and $j$ in $\{1,\dots,\ell_{k}\}$.
Right-multiplying $\bq_{\tau_k + {j-1}}$ by $\mat{B}_{\tau_k+j}$ only
affects the entry at index $\tau_k+j$. It replaces 
$h_{\tau_k+j}$ by 
$$\sum_{i=1}^{j-1} q_{\{\tau_k+i,\tau_k+j,\dots,\tau_r\}}(z_{\tau_k+j}-z_{\tau_k+i}) \ \ - \ \ h_{\tau_k+j}.$$
Using the defining relation of divided differences, we get 
$$
q_{\{\tau_k+i,\tau_k+j,\dots,\tau_r\}}(z_{\tau_k+j}-z_{\tau_k+i}) = q_{\{\tau_k+i,\tau_k+j+1,\dots,\tau_r\}} - 
q_{\{\tau_k+j,\tau_k+j+1,\dots,\tau_r\}}.
$$
With the definition of $h_{\tau_k+j}$ in~\eqref{eqdef:hs}, the new entry 
at index $\tau_k+j$ simplifies 
as $-j q_{\{\tau_k+j,\tau_k+j+1,\dots,\tau_r\}}$. When we multiply the resulting 
vector by $\mat{C}_{\tau_k+j}$, we affect only entries
from indices $\tau_k+1$ to $\tau_k+j$. More precisely, 
the previous relation shows that
we obtain the vector
$$  \begin{pmatrix}
    q_{\{1, \tau_k+j, \dots, \tau_r\}} & \hdots &
    q_{\{\tau_k, \tau_k+j, \dots, \tau_r\}}
&
  q_{\{\tau_k+1, \tau_k+j+1, \dots, \tau_r\}}
&
\hdots 
&
  q_{\{\tau_k+j, \tau_k+j+1, \dots, \tau_r\}}
&
    h_{\tau_k+j+1} & \hdots & h_{\tau_r} 
    \end{pmatrix}.$$
Finally, right-multiplication by $\mat{D}_{\tau_k + j}$ affects 
entries of indices $1,\dots,\tau_k$. For $i=1,\dots,\tau_k$,
it replaces $q_{\{i,\tau_k+j,\dots,\tau_r\}}$ by
$$
q_{\{i,\tau_k+j,\dots,\tau_r\}} (z_{\tau_k+j}-z_i) +  q_{\{\tau_k+j,\tau_k+j+1,\dots,\tau_r\}}
=q_{\{i,\tau_k+j+1,\dots,\tau_r\}}.
$$
Thus, the resulting vector is 
$$  \begin{pmatrix}
    q_{\{1, \tau_k+j+1, \dots, \tau_r\}} & \hdots &
    q_{\{\tau_k, \tau_k+j+1, \dots, \tau_r\}}
&
  q_{\{\tau_k+1, \tau_k+j+1, \dots, \tau_r\}}
&
\hdots 
&
  q_{\{\tau_k+j, \tau_k+j+1, \dots, \tau_r\}}
&
    h_{\tau_k+j+1} & \hdots & h_{\tau_r} 
    \end{pmatrix}
    $$
which is precisely $\bq_{\tau_k+j}$, as claimed in~\eqref{eq:recQ}.

\section{Proof of Lemma~\ref{isolatedbound}}\label{app:bound}

To simplify our notation, for all $1 \leq s \leq \ell$, we abbreviate
$\eta_{\ell-s}(d-1, \dots, d-\ell)$ to $g_{\ell-s}$. Then, we claim
that one has $$g_{\ell-s} < d(d-1)\cdots
(d-\ell+1).$$ Indeed, let $f(t) = (t+d-1)(t+d-2)\cdots (t+d-\ell)$, so
that $f(1) = d(d-1)\cdots (d-\ell+1)$. From Vieta's formula we have
\[           
f(t) = \sum_{s= 0}^{\ell} g_{\ell-s} \, t^s 
\]            
and so we also have $f(1) = \sum_{s=0}^\ell
g_{\ell-s}$. Therefore, $$d(d-1)\cdots(d-\ell+1) = \sum_{s=0}^\ell
g_{\ell-s}$$ and so $g_{\ell-s} < d(d-1)\cdots (d-\ell+1)$ for
all $1 \leq s \leq \ell$.
 
Now, for any partition $\lambda =
(n_1^{\ell_1}\, \dots \, n_r^{\ell_r}) \vdash n$ of length $\ell_\lambda$, we
have
\begin{eqnarray*} \mathfrak{c}_\lambda & = &
d^s \frac{g_{\ell_\lambda-s}}{w_\lambda} \quad \text{with}\quad
w_\lambda = \prod_{i=1}^r \, \ell_i!\\
&=& d^s \frac{ \,\ell_\lambda!}{\prod_{i=1}^r \, \ell_i!} \, \frac{g_{\ell_\lambda -s}}{\ell_\lambda!}  \\
&=& d^s h(\lambda) \, \mathscr{F}_{d, \ell_\lambda,s},   
\end{eqnarray*}
where $h(\lambda) = \frac{\ell_\lambda!}{\prod_{i=1}^r \, \ell_i!} =
{{\ell_\lambda} \choose {\ell_1, \dots, \ell_r}}$ and $\mathscr{F}_{d,
  \ell_\lambda, s} = \frac{g_{\ell_\lambda - s}}{\ell_\lambda!}$. From our previous inequality we have 
$$\mathscr{F}_{d, \ell_\lambda, s} \leq \frac{d(d-1) \cdots (d-\ell_\lambda+1)}{\ell_\lambda!}
= {d \choose \ell_\lambda}$$ 
and so
\begin{equation} \label{eqqq0}
  \sum_{\lambda \vdash n, \, \ell_\lambda \geq s} \, \mathfrak{c}_\lambda \leq d^s
  \left( \sum_{\lambda  \vdash n, \, \ell_\lambda \geq s} \,
  h(\lambda) \, {d \choose \ell_\lambda}\right). 
  \end{equation}

Let $\bf{a}$ be a sequence of $m+1$ numbers $(a_0, a_1, \dots, a_m)$
and let $p_{{\bf a}}(t) = \sum_{i=0}^m a_i \, t^i$ be its generating
polynomial. The \textit{polynomial coefficients} associated to
$\bf{a}$ are defined by
\[
{k \choose n}_{\bf{a}} = 
\begin{cases}
[t^n] \, (p_{\bf a}(t)^k) , & \text{ if } 0\leq n\leq mk\\
0 , & \text{ if } n < 0 \ \text{or} \ n > mk
\end{cases}
\] 
where $[t^n] \, \sum_i c_i t_i = c_n$ is the coefficient of $t^n$ in
the series $\sum_i c_i t_i$.  For any partition $\lambda$ of $n$, let
further $\lambda'$ be its conjugate partition. By~\cite[Lemma~2.1]{Fah12}, we have
\begin{equation} \label{coeff}
{k \choose n}_{\bm{a}} = \sum_{\substack{\lambda \vdash n,\\
    \ell_{\lambda'} \leq n}} a_0^{k-\ell_{\lambda'}} \, h(\lambda) w_{\bm a}(\lambda) {k
  \choose \ell_{\lambda}},
\end{equation} 
where $w_{\bm a}(\lambda)$ is the function $w_{\bm a}(\lambda) = \prod_{i=1}^ma_i^{\ell_i}$,
and $\ell_\lambda,\ell_{\lambda'}$ are the respective lengths of $\lambda$ and $\lambda'$.
If we consider $m = n$, ${\bm a} = (1, \dots, 1) = {\bm 1}$ and $k = d$, then equation~(\ref{coeff}) becomes
\[
{d \choose n} _{\bm{1}} =  \sum_{\substack{\lambda \vdash n,\\ \ell_{\lambda'}
    \leq n}}  h(\lambda) {d \choose \ell_{\lambda'}}. 
\] 
For any partition $\lambda$ of $n$,  the length of its conjugate satisfies $\ell_{\lambda'} \leq n$ and so
\begin{equation} \label{eqqq}
[t^n] (1 + t + \cdots + t^n)^d = {d \choose n} _{\bm{1}} =
\sum_{\substack{\lambda \vdash n}}  h(\lambda) {d \choose \ell_\lambda}.
\end{equation}
Furthermore, 
\begin{eqnarray*}
(1 + t + \cdots + t^n)^d &=&  (1-t^{n+1})^d \, (1-t)^d = \Big(
                             \sum_{k=0}^d \, (-1)^k {d \choose k} \,
                             t^{(n+1)k} \Big) 
    \Big(  \sum_{i=0}^{\infty} {{d + i -1}\choose i} t^i \Big),
\end{eqnarray*}
where $t^n$ appears only when $k=0$ and $i=n$. In other words, 
\begin{equation} \label{coeff_bound}
[t^n] \, (1 + t + \cdots + t^n)^d = {{n+d-1} \choose {n}}. 
\end{equation}
Combining~(\ref{eqqq0}), (\ref{eqqq}) and (\ref{coeff_bound}), gives
\[
  \sum_{\lambda \vdash n, \, \ell_\lambda \geq s} \, \mathfrak{c}_\lambda  \leq
  d^s \, \left( \sum_{\lambda  \vdash n} \, h(\lambda) \, {d
  \choose \ell_\lambda}\right)  \leq  d^s {{n+d-1} \choose
  {n}} .
\]
We prove the inequality $ \sum_{\lambda \vdash n, \, \ell_\lambda \geq
  s} \, \mathfrak{e}_\lambda \le n (d+1)^s {{n+d} \choose {n}}$
similarly.

\end{document}